%% file: main_arxiv.tex
\documentclass[fleqn,10pt]{wlscirep}
\usepackage[utf8]{inputenc}
\usepackage[T1]{fontenc}

\usepackage{datetime}
  \newdateformat{monthonly}{\monthname[\THEMONTH]}
\usepackage{amssymb,amsthm,amsmath}
\usepackage{graphicx}
\usepackage{floatrow}
\usepackage{booktabs}
\usepackage{paralist}
\usepackage{titlesec}
\usepackage{parskip}
\titleformat{\chapter}[display]
{\bfseries\huge}
{\filleft\Large\chaptertitlename~\thechapter}
{3ex}
{\titlerule\vspace{1.5ex}\filright}
[\vspace{1ex}\titlerule]

\usepackage{tikz-cd}
\usepackage{tikz-network}
\usepackage{caption}
\usepackage{subcaption}
\usepackage{mathdots}
\usepackage{csquotes}
\usepackage[linesnumbered,ruled,vlined,noend]{algorithm2e}
\usepackage{todonotes}
\usepackage{bm}
\usepackage{xargs}
\usepackage{comment}
\usepackage{standalone}
\usepackage{sgame}

\tikzset{ball/.style={circle, draw, fill=black,inner sep=0pt, minimum width=4pt}}
\tikzset{nd/.style={inner sep=1pt}}
\tikzset{CRS/.style={circle, draw,inner sep=1pt, minimum width=8pt}}

\usepackage{pgfplots}
\pgfplotsset{compat = newest}

\usepackage{xcolor}
\usetikzlibrary{decorations.markings}
\usetikzlibrary{arrows.meta}
\tikzset{>=Latex}
\tikzset{
  set arrow inside/.code={\pgfqkeys{/tikz/arrow inside}{#1}},
  set arrow inside={end/.initial=>, opt/.initial=},
  /pgf/decoration/Mark/.style={
    mark/.expanded=at position #1 with
    {
      \noexpand\arrow[\pgfkeysvalueof{/tikz/arrow inside/opt}]{\pgfkeysvalueof{/tikz/arrow inside/end}}
    }
  },
  arrow inside/.style 2 args={
    set arrow inside={#1},
    postaction={
      decorate,decoration={
        markings,Mark/.list={#2}
      }
    }
  },
}

\theoremstyle{plain}
\newtheorem{thm}{Theorem}[section] 

\theoremstyle{definition}
\newtheorem{defn}[thm]{Definition} 
\newtheorem{exmp}[thm]{Example} 
\newtheorem{lem}[thm]{Lemma}
\newtheorem{corol}[thm]{Corollary}

\newcommand{\real}{\mathbb{R}}

\newcommand{\arc}[3][]{\begin{tikzcd} #2 \ar[r,-Latex,"#1"] \pgfmatrixnextcell #3 \end{tikzcd}}
\newcommand{\edge}[3][]{\begin{tikzcd} #2 \ar[r,dash,"#1"] \pgfmatrixnextcell #3 \end{tikzcd}}

\DeclareMathOperator{\pathweight}{pathweight}

\title{The graph structure of two-player games}

\author[1]{Oliver Biggar}
\author[1]{Iman Shames}
\affil[1]{CIICADA Lab, Australian National University, Canberra, 2601, Australia}

\affil[*]{oliver.biggar@anu.edu.au}

\begin{abstract}
In this paper we analyse two-player games by their \emph{response graphs}. The response graph has nodes which are strategy profiles, with an arc between profiles if they differ in the strategy of a single player, with the direction of the arc indicating the preferred option for that player. Response graphs, and particularly their \emph{sink strongly connected components}, play an important role in modern techniques in evolutionary game theory and multi-agent learning. We show that the response graph is a simple and well-motivated model of strategic interaction which captures many non-trivial properties of a game, despite not depending on cardinal payoffs. We characterise the games which share a response graph with a zero-sum or potential game respectively, and demonstrate a duality between these sets. This allows us to understand the influence of these properties on the response graph. The response graphs of Matching Pennies and Coordination are shown to play a key role in all two-player games: every non-iteratively-dominated strategy takes part in a subgame with these graph structures. As a corollary, any game sharing a response graph with both a zero-sum game and potential game must be dominance-solvable. Finally, we demonstrate our results on some larger games.
\end{abstract}
\begin{document}

\flushbottom
\maketitle
\thispagestyle{empty}

\section{Introduction}

One of the most fundamental questions in game theory is that of \emph{representing preference}\cite{pareto1919manuale,von2007theory,rasmusen1989games,myerson1997game}: how should we model the preferences of players over their strategies? The established solution, originating in Von Neumann and Morgenstern's axiomatisation of \emph{utility}\cite{von2007theory}, is to assign to each player a \emph{real-valued payoff}, for each combination of strategies. Soon afterwards, John Nash invented his eponymous equilibrium concept\cite{nash1951non}, which he proved exists in every game modelled by Von Neumann-Morgenstern utility. This elegant result established the Nash equilibrium as a clear choice of the outcome of a game. Importantly, these two concepts are \emph{mutually reinforcing}: Von Neumann-Morgenstern utility lays the mathematical foundation to prove the existence of Nash equilibria, and the existence of Nash equilibria retrospectively justifies the choice of the Von Neumann-Morgenstern model. Together, this began a flurry of game-theoretic research which cemented both Von Neumann-Morgenstern utility and the Nash equilibrium as central notions in economic thought\cite{myerson1997game}.

Unfortunately, many games do not have obvious choices of utility values. Because of this, other game models---such as \emph{ordinal games}\cite{mertens2004ordinality,durieu2008ordinal,cruz2000ordinal}---have persisted as alternatives which make weaker assumptions on what we, as modellers, must know about a strategic interaction we intend to analyse. But these models have been hindered by the dominance of the Nash equilibrium in the game theory literature\cite{myerson1997game}; without a solution concept as clear and compelling as the Nash equilibrium, such models have been unable to overtake the prevailing Von Neumann-Morgenstern approach.

However, as game theory has grown to be a significant tool in biology\cite{smith1973logic}, computer science\cite{roughgarden2010algorithmic} and multi-agent learning\cite{shoham2008multiagent} the Nash equilibrium has been found to be a less compelling solution concept than was once thought. The first argument comes from computational complexity: Nash equilibria are intractable to compute from the description of the game\cite{daskalakis2009complexity}, even in two-player games\cite{chen2006settling}. Neither we, the analysts, nor the players themselves, can feasibly compute Nash equilibrium strategies. The second argument comes from \emph{evolutionary game theory}, the subfield containing population dynamics and learning\cite{sandholm2010population}. A series of results have established that evolution or learning rules do not\cite{sandholm2010population,piliouras2014optimization,papadimitriou2016nash,vlatakis2020no} and generally \emph{cannot}\cite{benaim2012perturbations,hart2003uncoupled} converge to Nash equilibria. Instead, non-equilibrium behaviour is the rule rather than the exception, giving the Nash equilibrium relatively little predictive value~\cite{papadimitriou2019game,omidshafiei2019alpha,piliouras2014optimization,kleinberg2011beyond,cheung2019vortices}.

Recently, great advances have been made in AI\cite{silver2016mastering,silver2017mastering,silver2018general}, using ideas from multi-agent learning\cite{hernandez2019survey,yang2020overview}. If the Nash equilibrium is not a satisfactory notion of outcome in these fields, we are motivated to seek new approaches to evolutionary game theory which can explain and shed light on learning\cite{kleinberg2011beyond,papadimitriou2016nash,papadimitriou2019game,omidshafiei2019alpha}. In particular, if our model no longer requires the Nash equilibrium, we are free to consider models which do not use Von Neumann-Morgenstern utility.

Our approach is based on \emph{Occam's razor}, the principle that the simplest model capable of describing the concept of interest is often the best. \emph{We find a solution that is both computable and aligns with the outcome of evolutionary processes by simplifying our model of player preferences}.

\begin{table}
    \centering
    \begin{game}{3}{3}[Player~1][Player~2]
         & Rock & Paper & Scissors \\
         Rock & $0,0$ & $-1,1$ & $1,-1$ \\
         Paper & $1,-1$ & $0,0$ & $-1,1$ \\
         Scissors & $-1,1$ & $1,-1$ & $0,0$ 
    \end{game}
    \caption{A standard presentation of the Rock-Paper-Scissors game.\cite{sandholm2010population}}
    \label{tab:rps}
\end{table}

\begin{figure}
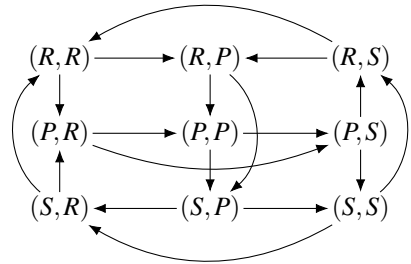

    \centering
    \includestandalone{figs/3x3/RPS_labeled}
    \caption{The response graph of Rock-Paper-Scissors.}
    \label{fig:rps table graph}
\end{figure}
\begin{figure}
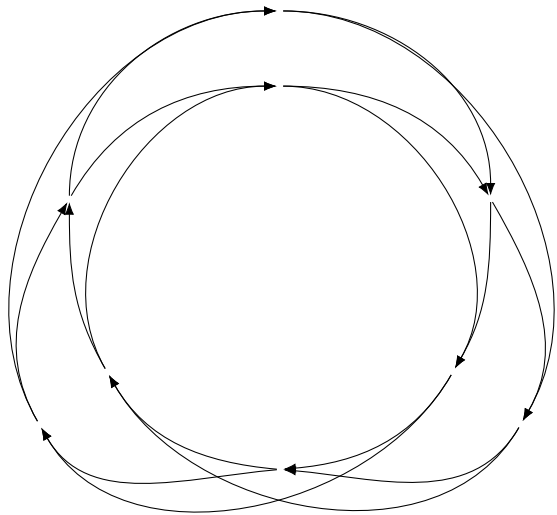

    \centering
    \includestandalone{figs/3x3/RPS}
    \caption{An alternate presentation of the response graph of Rock-Paper-Scissors, emphasising its M\"obius strip structure.}
    \label{fig:rps round}
\end{figure}

We begin with a concrete example of a game: Rock-Paper-Scissors. In this game, each player simultaneously chooses one of `Rock', `Paper' or `Scissors', where Rock defeats Scissors, Scissors defeats Paper, Paper defeats Rock, and playing the same option yields a tie. If we were explaining the game to another person, this description (along with the standard assumption that winning is preferred to a tie which is itself preferred to losing) is sufficient. Intuition tells us this description should also be sufficient to analyse the game mathematically. Yet, the payoffs have not been specified---this is not a game in the sense of game theory\cite{myerson1997game}. Two "Rock-Paper-Scissors" games obeying these constraints and yet differing in the payoff values are different games. We have specified only the \emph{preference order} over each player's strategies, given a fixed choice of strategy for the other player. For instance, if Player~1 plays Rock, we know that Player~2 prefers their options in the best-to-worst order: Paper, Rock, Scissors. This is the underlying structure of Rock-Paper-Scissors---if I reward the winner of the game with \$2 instead of \$1, it should not become a different game! Indeed any $3\times 3$ game with these preference orders, even without payoffs specified, will generally be referred to as "Rock-Paper-Scissors". For the same reason, Rock-Paper-Scissors is often \cite{sandholm2010population} presented with the "default" 1, 0, -1 payoffs, as in Table~\ref{tab:rps}. These payoffs are serving only to instantiate the preference orders. So, can we cut out the intermediary, and define a game by the preference orders alone?

These preference orders are captured precisely in an object called the \emph{response graph} of a game\cite{candogan2011flows,papadimitriou2019game}. The nodes of this graph are the strategy profiles, and there is an arc between profiles if they differ in the strategy of a single player, with arcs directed toward the preferred profile for that player. The response graph of Rock-Paper-Scissors is shown in Figure~\ref{fig:rps table graph} and again in Figure~\ref{fig:rps round}, with the latter emphasising the symmetric cycle structure. We can present Figure~\ref{fig:rps round} without labelling the nodes by profiles---the profiles can be reconstructed from the graph in linear time up to renaming (Theorem~\ref{reconstruction}), so the response graph implicitly handles the problem of renaming players and strategies. Importantly for our purposes, response graphs play an important part in modern developments in machine learning and evolutionary game theory~\cite{papadimitriou2019game,papadimitriou2018nash,omidshafiei2019alpha,kleinberg2011beyond,papadimitriou2016nash}. A key concept is the \emph{sink strongly connected components} of the response graph (which we shall shorten to \emph{sink components}), which are a solution concept generalising pure Nash equilibria. Recent work\cite{biggarthesis} has shown that under the replicator dynamic, a common choice of evolutionary dynamic\cite{sandholm2010population}, the sink components are contained in \emph{sink chain components}, a topological concept which emerges from the Fundamental Theorem of Dynamical Systems\cite{conley1978isolated}. Sink chain components represent the `long-run' outcome of a dynamic process such as learning or evolution on a game. This result gives a compelling motivation for sink components as a dynamic---and unlike Nash equilibria, \emph{predictive}---solution concept for games. They are also tractable to compute\cite{papadimitriou2019game,omidshafiei2019alpha}. Building on these ideas, Omidshafiei et al\cite{omidshafiei2019alpha} present a new approach called \emph{$\alpha$-rank} for ranking the strength of agents in multi-agent settings using the response graph and sink components. When applied to a `biased' Rock-Paper-Scissors game with differing payoffs in different profiles the authors find that $\alpha$-rank still gives an equal ranking to each strategy `Rock', `Paper' and `Scissors'\cite{omidshafiei2019alpha}, suggesting the long-run strength of these strategies is a property of the response graph. A variant of the response graph also exists, called the \emph{weighted response graph}, where arcs are weighted by the difference in payoff for the associated player. Weighted response graphs provide a mechanism to decompose games\cite{candogan2011flows,candogan2013dynamics,hwang2020strategic} up to \emph{strategic equivalence}.  More recently, the spectrum\cite{diestel2016graph} of the response graph has been used to describe the topological landscape of multiplayer games\cite{omidshafiei2020navigating} for the purposes of analysing and comparing games. Sink components have also been used\cite{goemans2005sink} as an alternative measure of the Price of Anarchy\cite{roughgarden2005selfish}.

The response graph is defined by the preference orders; it does not depend closely on payoffs. If two games have different payoffs but the preference order for any given player is equal for any fixed choice of strategies for the other players, the response graphs are the same. The response graph is more general than an \emph{ordinal game}\cite{mertens2004ordinality,durieu2008ordinal,barany1992classification}---in that model, two games are \emph{ordinal-equivalent} if each player's order over \emph{all} profiles is the same. In the Rock-Paper-Scissors example, this would require modelling whether Player~1 prefers winning in the profile (Rock, Scissors) to winning in the profile (Scissors, Paper), even though they can never unilaterally choose between these profiles! The notion of \emph{strategic equivalence}\cite{morris2004best,candogan2011flows,hwang2020strategic} takes this into account, defining two games to be \emph{strategically-equivalent} if the relative payoff difference between \emph{comparable} profiles---those differing in only one player---is equal. Strategic equivalence is captured by the \emph{weighted} response graph; it is motivated by the fact that the Nash equilibrium is invariant under strategic equivalence \cite{candogan2011flows,candogan2013dynamics}. In fact, strategic equivalence is defined by the preference orders over all \emph{mixed} profiles~\cite{moulin1978strategically}. Unlike ordinal equivalence, strategic equivalence \emph{does} depend on the cardinal value of payoffs. The (unweighted) response graph combines the strengths of both equivalences, generalising ordinal and strategic games into a simple and well-motivated model capturing the `underlying structure' of a game\cite{biggarthesis}.

The response graph is a simple and general model. If it is to be a \emph{good} model, by Occam's razor, it must also be capable of describing non-trivial properties of a game. \textbf{That is the goal of this paper}: to establish that, despite their generality and combinatorial nature, response graphs capture important and non-trivial game-theoretic properties which extend the existing theory of two-player games. Though we focus on two-player games, we expect that the response graph approach will be equally applicable for general games. This line of inquiry allows us to conceptually separate those properties of a game which are defined by the payoffs from those which are defined by the preferences alone, and so provide a better-informed theory in applications where we cannot reliably model real-valued payoffs. Recalling the connection\cite{papadimitriou2019game,omidshafiei2019alpha,biggarthesis} between the sink components and the long-run outcome of the replicator dynamics on a game, we find that investigating the sink components, as we do in this paper, sheds light on evolution and learning.

\subsection{Contributions}

In this paper we study the response graphs of two-player games. To isolate the influence of the response graph, we study games \emph{modulo isomorphism of response graphs}. That is, we say two games are \emph{preference-equivalent} if their response graphs are isomorphic. It is easy to see that pure Nash equilibria and strict iterated dominance of pure strategies are properties which are invariant under this equivalence relation. We define the \emph{preference-zero-sum} and \emph{preference-potential} game as those games which are preference-equivalent to either a zero-sum game\cite{von2007theory} or potential game\cite{monderer1996potential} respectively. These classes of games are particularly important in game dynamics\cite{sandholm2010population}, and so understanding their sink components is a natural question of interest. Zero-sum two-player games particularly are one of the most well-studied classes of games\cite{von2007theory}, and their definition depends crucially on payoffs. Despite this, we find that key properties of zero-sum games extend to the much broader set of preference-zero-sum games, showing that being zero-sum is to some degree a graph property. While the preference-potential games are known to be those with \emph{acyclic} response graphs~\cite{durieu2008ordinal}, the preference-zero-sum games have---to our knowledge---never been characterised.
We prove that a two-player game is preference-zero-sum if and only if it is acyclic after \emph{reflection}, which is a reversal of preferences for one player (Corollary~\ref{corol: preference duality}). Thus the graph property underlying the zero-sum property is acyclicity. We find that the existence of pure Nash equilibria in potential games extends to preference-potential games, and the uniqueness of Nash equilibria in generic zero-sum games translates to uniqueness of the sink component in generic preference-zero-sum games (Lemma~\ref{one sink}).

The Matching Pennies and Coordination\cite{sandholm2010population} $2\times 2$ games respectively form the prototypical examples of zero-sum and potential games. Their response graphs (Figure~\ref{fig:p MP} and~\ref{fig:p CO}) are the 4-cycle and the \emph{reflected} 4-cycle (Definition~\ref{def: reflection}). Remarkably, we find that these two graphs play a \emph{fundamental role} in bringing strategic complexity to \emph{all two-player games}. First, any two-player game which has multiple sink components must contain the response graph of CO as an induced subgraph (Theorem~\ref{one sink}). Second, in any two-player game, every non-iteratively-dominated strategy takes part in a $2\times 2$ subgame whose graph is that of Matching Pennies or Coordination (Theorem~\ref{ CO and MP theorem}). As a consequence we find a new game theory result: if all $2\times 2$ subgames of a two-player game have a dominated strategy, then the game is dominance-solvable. Combining this result with the previous characterisations, we obtain the surprising corollary that any two-player game both preference-zero-sum and preference-potential must be dominance-solvable (Corollary~\ref{dominance theorem}). These results are far-reaching, because the classes of preference-zero-sum and preference-potential games are very broad; every $2\times 2$ game is either preference-zero-sum, preference-potential, or both, in which case it is dominance-solvable. Even among $2\times 3$ games, there is \emph{only one} generic response graph which is neither preference-zero-sum nor preference-potential (Figure~\ref{fig:2x3 MPCO}).

Finally, we demonstrate our results by exploring $2\times 3$, $2\times 4$ and $3\times 3$ response graphs. We show how the techniques of the paper allow us to reason about such games easily, and we construct a stock of examples with interesting properties. For instance, we construct a generic $3\times 3$ preference-zero-sum game with a pure Nash equilibrium (Figure~\ref{fig:inner}) and show that it is the unique response graph with these properties.

The proofs can be found in the Supplementary Material.

\section{Preliminaries} \label{sec: preliminaries}

A \emph{graph}\cite{bang2008digraphs} is a pair $G = (N,A)$, where $N$ is a finite set of \emph{nodes} and $A\subseteq N\times N$ is a finite set of \emph{arcs}. We depict an arc $(x,y)\in A$ by $\arc{x}{y}$. If for some nodes $x$ and $y$ we have both $(x,y)\in A$ and $(y,x)\in A$ then we refer to this pair of arcs collectively as an \emph{undirected edge}, and depict it as $\edge{x}{y}$. If $(x,y)\in A$ implies $(y,x)\in A$ for any pair of nodes, then all arcs are undirected edges, and we call $G$ an \emph{undirected graph}. Each graph $G$ has an associated undirected graph $G'$, called the \emph{underlying graph}, given by requiring that for each arc $(x,y)$ in $G$ there are arcs $(x,y)$ and $(y,x)$ in $G'$. Removing one of the arcs $(x,y)$ or $(y,x)$ from an undirected edge $\edge{x}{y}$ gives a standard arc, a process we call \emph{orienting} the undirected edge. An \emph{orientation} of an undirected graph is any graph formed by oriented some of its undirected edges. 
A \emph{path} is a sequence $v_1,v_2,\dots,v_n$ of distinct nodes where there is an arc $\arc{v_i}{v_{i+1}}$ for every $i$ in $1,2,\dots,n-1$. An \emph{undirected path} is a path in the underlying graph. If there is also an arc $\arc{v_n}{v_1}$, we call this a \emph{cycle}. A graph with no cycles is called \emph{acyclic}. If there is a path from a node $v$ to a node $w$ we say \emph{$w$ is reachable from $v$}. Reachability defines a preorder on the nodes of a graph. Two nodes are equivalent under this preorder if both are reachable from each other. The equivalence classes of this relation are called the \emph{strongly connected components}. The minimal elements of this order we call the \emph{sink components}.
For any subset $X\subseteq N$ of nodes, there is an associated graph given by including exactly the arcs between nodes in $X$. This is called the \emph{subgraph induced by $X$} or simply an \emph{induced subgraph}. Two graphs $(N_1,A_1)$ and $(N_2,A_2)$ are \emph{isomorphic} if there is a map $\varphi : N_1\to N_2$ where $\arc{v}{w}\in A_1$ if and only if $\arc{\varphi(v)}{\varphi(w)}\in A_2$.

All games in this paper are two-player normal-form games with finite strategy sets\cite{myerson1997game}. Such a game is defined by a pair of \emph{payoff functions} $u_1,u_2 : S_1\times S_2 \to \real$, where $S_1$ and $S_2$ are finite sets, called the \emph{strategy sets}, whose elements are \emph{strategies}. If $|S_1| = n$ and $|S_2| = m$, we call the game a $n\times m$ game. A \emph{strategy profile} is a pair $(s_1,s_2)\in S_1\times S_2$. We call $u_1(s_1,s_2)$ the \emph{payoff} to player 1 in the profile $(s_1,s_2)$. Two profiles are \emph{$i$-comparable} if they differ in the strategy of player $i$ only, and are \emph{comparable} if they are $i$-comparable for some $i$. We say that a strategy $s\in S_1 $ \emph{dominates} a strategy $t\in S_1$ if $u_1(s,r) > u_1(t,r)$ for every strategy $r\in S_2$, and the same definition holds analogously for player 2. The strategy $t$ is called \emph{dominated}. If we delete some dominated strategy (forming the subgame given by removing this strategy), other strategies can become dominated in the new game. This process is called \emph{iterated elimination of dominated strategies}\cite{myerson1997game}. Any strategy deleted during this process is called \emph{iteratively dominated}, and otherwise a strategy is said to \emph{survive iterated dominance}. If a game has only one profile that survives iterated dominance, then that profile is the unique pure Nash equilibrium, and we call the game \emph{dominance-solvable}. While \emph{mixed strategies} can also dominate strategies\cite{myerson1997game}, we focus here on the case where all strategies are pure.

The \emph{response graph} of the game is the graph whose node set is $S_1\times S_2$, with an arc $\arc{(s_1,s_2)}{(t_1,t_2)}$ if the profiles $(s_1,s_2)$ and $(t_1,t_2)$ are $i$-comparable and $u_i(t_1,t_2) \geq u_i(s_1,s_2)$. The \emph{weighted response graph} (called the \emph{game graph} in \cite{candogan2011flows}) has the additional property that the arc $\arc{(s_1,s_2)}{(t_1,t_2)}$ is weighted by the non-negative number $u_i(t_1,t_2) - u_i(s_1,s_2)$. If $u_i(t_1,t_2) = u_i(s_1,s_2)$ then say player $i$ is \emph{indifferent} between $(s_1,s_2)$ and $(t_1,t_2)$, and there are arcs in both directions, that is, there is an undirected edge $\edge{u_i(t_1,t_2)}{u_i(s_1,s_2)}$.  In the weighted response graph, undirected edges are weighted by zero. 
 A \emph{subgame} of a game is the game given by restricting $u_1$ and $u_2$ to the domain $T_1\times T_2$, where $T_1\subseteq S_1$ and $T_2\subseteq S_2$. A pure Nash equilibrium is a profile where all $i$-comparable profiles give player $i$ no improvement in payoff, for any $i$. Equivalently, $(s_1,s_2)$ is a pure Nash equilibrium if and only if for every comparable profile $(t_1,t_2)$ there is an arc $\arc{(t_1,t_2)}{(s_1,s_2)}$ in the response graph.
The sink components of the response graph have also been called \emph{Markov-Conley chains}\cite{papadimitriou2019game,omidshafiei2019alpha}, but in that context they were augmented with the structure of a Markov chain.

In game theory, a property of a game is \emph{generic} if almost all games in payoff space possess the property\cite{fudenberg1991game}. We shall focus one generic property in particular; specifically, the absence of undirected edges. We shall call a game \emph{generic} if the payoffs to player $i$ in two $i$-comparable profiles are never equal---that is, if its response graph has no undirected edges.

\begin{defn}
Two two-player games are \emph{preference-equivalent} if their response graphs are isomorphic. They are \emph{strategically-equivalent}\cite{candogan2011flows} if their \emph{weighted} response graphs are also isomorphic. 
\end{defn}

We observe first that strategic equivalence implies preference equivalence. Secondly, note that the graph isomorphism criterion implicitly handles renaming of strategies and reordering of players. As an example, the game $(u_1,u_2)$ and $(u_2,u_1)$ are strategically equivalent, because the map $\varphi: S_1\times S_2 \to S_2\times S_1$, $\varphi(a,b) = (b,a)$ defines an isomorphism of the weighted response graphs. While our focus is on preference-equivalence, we do make use of the more restrictive notion of strategic equivalence. Unlike preference equivalence, strategic equivalence has been well-studied in game theory\cite{morris2004best,hwang2020strategic,hwang2020simple,candogan2011flows,candogan2013dynamics} because Nash equilibria are invariant under strategic equivalence\cite{candogan2011flows}.

\section{Graphs from Games}

To motivate our thinking about response graphs, we begin by considering the $2\times 2$ generic games, the simplest non-trivial games. While there are infinitely many such games, there are only four non-isomorphic response graphs, which we call Matching Pennies (MP), Coordination (CO), Single-dominance (SD) and Double-dominance (DD). We can deduce this by brute force: the underlying graph of any $2\times 2$ response graph is an undirected 4-cycle, and there are four distinct orientations of this graph, shown in Figure~\ref{fig: 2x2 response graphs}.
\begin{figure}
    \centering
    \begin{subfigure}{0.22\textwidth}
    \centering
    \begin{tikzpicture}
    \node[nd] (A) at (0,0) {};
    \node[nd] (B) at (0,1.5) {};
    \node[nd] (C) at (1.5,0) {};
    \node[nd] (D) at (1.5,1.5) {};
    \draw[->] (A) to (B);
    \draw[->] (A) to (C);
    \draw[->] (C) to (D);
    \draw[->] (B) to (D);
    \end{tikzpicture}
    \caption{Double-dominance}
    \label{fig:p DD}
    \end{subfigure}
    \begin{subfigure}{0.22\textwidth}
    \centering
    \begin{tikzpicture}
    \node[nd] (A) at (0,0) {};
    \node[nd] (B) at (0,1.5) {};
    \node[nd] (C) at (1.5,0) {};
    \node[nd] (D) at (1.5,1.5) {};
    \draw[->] (A) to (B);
    \draw[->] (B) to (D);
    \draw[->] (A) to (C);
    \draw[->] (D) to (C);
    \end{tikzpicture}
    \caption{Single-dominance}
    \label{fig:p SD}
    \end{subfigure}
    \begin{subfigure}{0.22\textwidth}
    \centering
    \begin{tikzpicture}
    \node[nd] (A) at (0,0) {};
    \node[nd] (B) at (0,1.5) {};
    \node[nd] (C) at (1.5,0) {};
    \node[nd] (D) at (1.5,1.5) {};
    \draw[->] (A) to (B);
    \draw[->] (A) to (C);
    \draw[->] (D) to (B);
    \draw[->] (D) to (C);
    \end{tikzpicture}
    \caption{Coordination}
    \label{fig:p CO}
    \end{subfigure}
    \begin{subfigure}{0.22\textwidth}
    \centering
    \begin{tikzpicture}
    \node[nd] (A) at (0,0) {};
    \node[nd] (B) at (0,1.5) {};
    \node[nd] (C) at (1.5,0) {};
    \node[nd] (D) at (1.5,1.5) {};
    \draw[->] (B) to (A);
    \draw[->] (A) to (C);
    \draw[->] (C) to (D);
    \draw[->] (D) to (B);
    \end{tikzpicture}
    \caption{Matching Pennies}
    \label{fig:p MP}
    \end{subfigure}
    \caption{The four non-isomorphic response graphs of generic $2\times 2$ games.}
    \label{fig: 2x2 response graphs}
\end{figure}
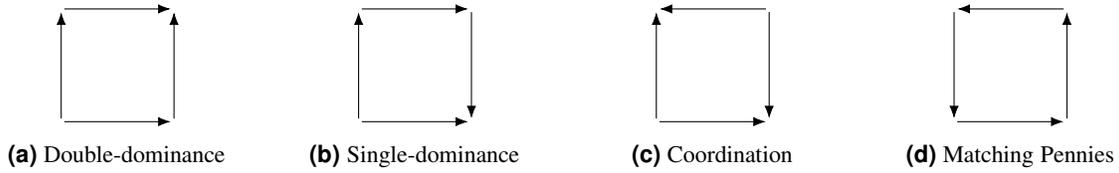

The Matching Pennies and Coordination graphs are named for some well-known games of the same name\cite{sandholm2010population,papadimitriou2018nash}. The Single- and Double-dominance graphs are named for the fact that they have one or two dominated strategies, respectively. These games showcase the influence of the response graph on two-player games: any game whose response graph is SD or DD is dominance-solvable; a game whose response graph is CO has two pure Nash equilibria; a game whose response graph is MP has no pure Nash equilibria. Any generic $2\times 2$ game possesses one of these response graphs. A \emph{non-generic} $2\times 2$ game has a response graph where some of these arcs are undirected edges; an example is shown in Figure~\ref{fig:weak MP}. We can fit such graphs into our classification with the notion of a \emph{weak form}.
\begin{defn}[Weak Form]
A graph $G$ is a \emph{weak form} of another graph $H$ if $H$ is an orientation of $G$.
\end{defn}

\begin{figure}
    \centering
    \begin{tikzpicture}
    \node[nd] (A) at (0,0) {};
    \node[nd] (B) at (0,1.5) {};
    \node[nd] (C) at (1.5,0) {};
    \node[nd] (D) at (1.5,1.5) {};
    \draw[->] (A) to (B);
    \draw (A) to (C);
    \draw[->] (D) to (C);
    \draw[->] (B) to (D);
    \end{tikzpicture}
    \caption{A graph that is weak MP and weak SD.}
    \label{fig:weak MP}
\end{figure}
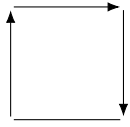

As an example, Figure~\ref{fig:weak MP} is a weak form of both SD and MP, as orienting the undirected edge gives either MP or SD. We say a graph $G$ \emph{contains} a graph $H$ if $H$ is an induced subgraph of $G$. We show later that weak forms of MP and CO are contained in all non-dominated two-player games.
There is an important fact to note in our presentation here. Unlike in Figure~\ref{fig:rps table graph}, where we labelled each node in the response graph by the associated profile, the graphs in Figure~\ref{fig: 2x2 response graphs} are not labelled by profiles. It turns out that this does not matter: if a graph is a response graph, the profiles can be recovered uniquely up to renaming of strategies.

\begin{thm} \label{reconstruction}
Given a graph $G$, we can construct a game whose response graph is $G$, or determine that no such games exist, in time linear in the number of arcs.
\end{thm}
This theorem follows from the fact that the underlying graphs of response graphs are \emph{Hamming graphs}\cite{hammack2011handbook}. It tells us that the graph structure is alone sufficient to analyse the preference orders in the game. Further, this allows us to represent \emph{implicitly} the independence of the game from renaming of strategies or reordering of players. Consequently, we can present our graphs in the natural graph-theoretic way (up to isomorphism) without losing any game-theoretic information. Consider Figures~\ref{fig:rps table graph} and \ref{fig:rps round}, both of which depict the response graph of Rock-Paper-Scissors. While Figure~\ref{fig:rps table graph} mimics the payoff table structure of Table~\ref{tab:rps}, Figure~\ref{fig:rps round} makes clear the symmetric M\"obius-strip-like structure of the graph, which is otherwise obscure. In much the same way, we find that the presence of subgames with the structure of MP or CO can also be expressed graph-theoretically.

\begin{lem} \label{lem: subgame lemma}
If the response graph of a two-player game contains the response graph of a $2\times 2$ game, then the profiles which take part form a $2\times 2$ subgame.
\end{lem}

In particular, every appearance of MP or CO in a response graph occurs in four profiles which make up a subgame of the associated game. Hence we can interchangeably use `the response graph contains MP' and `the game has a $2\times 2$ subgame whose response graph is isomorphic to MP', because these statements mean the same thing.

\section{Two-Player Zero-Sum and Potential Duality}

In this section we discuss two famous classes of games: zero-sum games~\cite{von2007theory} and potential games~\cite{monderer1996potential}. We characterise these classes up to preference-equivalence; currently they have only been characterised up to the more restrictive notion of strategic equivalence\cite{hwang2020strategic}. Generic preference-potential games turn out to be precisely those whose response graphs are acyclic (this is straightforward to prove, and follows from results in\cite{cruz2000ordinal}). Using a relationship between strategically-potential and strategically-zero-sum games, we establish a duality between preference-potential and preference-zero-sum games, and use this to characterise the generic preference-zero-sum games as the \emph{reflected} acyclic games.

\begin{defn} \label{def: potential games}
A two-player game $(u_1,u_2)$ is called a \emph{potential game}\cite{monderer1996potential} if there is a function $\phi:S_1\times S_2 \to\real$ such that for every pair of $i$-comparable profiles $p$ and $q$, $\phi(p) - \phi(q) = u_i(p) - u_i(q)$. A game is \emph{preference-potential} if it is preference-equivalent to some potential game. It is \emph{strategically-potential} if it is strategically-equivalent to some potential game.
\end{defn}

That is, the relative payoffs to each player can be defined by a single real-valued function, named the \emph{potential function}, by analogy with physics. There, a dynamic $f$ is called \emph{potential} if $f = \nabla \varphi$, where $\varphi$ is a real-valued function. This means that $f$ is a gradient vector field. As we know from vector calculus~\cite{hubbard2015vector}, such vector fields are exactly those that are \emph{conservative}. Additionally, the fundamental theorem of calculus holds, and so $f$ is \emph{path-independent}---that is, the path integral of $f$ is always the difference between the values of $\varphi$ at the endpoints. 
In game theory, potential games are notable because they guarantee the existence of a pure Nash equilibrium~\cite{monderer1996potential}. Intuitively, the existence of a potential function prevents cycles of preference. This idea is well-captured by the response graph---in fact, a generic game is preference-potential if and only if its response graph is acyclic (Corollary~\ref{corol: preference duality}).

\begin{defn}
A two-player game $u$ is \emph{zero-sum} if $u_1(s_1,s_2) + u_2(s_1,s_2) = 0$ for any strategies $s_1$ and $s_2$ for players 1 and 2 respectively. A two-player game is \emph{preference-zero-sum} if it is preference-equivalent to a zero-sum game. It is \emph{strategically-zero-sum} if it is strategically-equivalent to a zero-sum game.
\end{defn}

Intuitively, zero-sum games capture the notion that one player's gain is always the other player's loss. This model aligns closely with the recreational games from which game theory takes its name; there, if one player wins, the other must lose. From the graph perspective, this suggests that the preference orders of the players in a zero-sum game are never aligned. Hearing this, one might suspect that response graphs like CO do not occur in zero-sum games. This guess does turn out to be correct, and the insight gained leads to a characterisation of preference-zero-sum games.

\begin{exmp}[Coordination is not preference-zero-sum] \label{coordination not zerosum}
Let $(u_1,u_2)$ be a $2\times 2$ zero-sum game, with payoffs $a$, $b$, $c$ and $d$ for player 1 in each of the four profiles, and their negations $-a$, $-b$, $-c$ and $-d$ as the payoffs to player 2. We assume for contradiction that this game has the response graph of Coordination. The setup is shown in Figure~\ref{fig:zerosum coordination}. To achieve this response graphs, the relative payoffs $c-d$, $b-a$, $a-c$ and $d-b$ must each be positive. However this implies that $c>d$, $b>a$, $a>c$ and $d>b$, giving the strict cycle $d > b > a > c > d$, which is impossible, and so we obtain a contradiction.
\begin{figure}
    \centering
    \begin{tikzpicture}
    \node (bl) at (0,0) {$(a,-a)$};
    \node (br) at (3,0) {$(b,-b)$};
    \node (tl) at (0,2) {$(c,-c)$};
    \node (tr) at (3,2) {$(d,-d)$};
    \draw[->] (bl) to node[above] {$b-a$} (br);
    \draw[->] (tr) to node[above] {$c-d$} (tl);
    \draw[->] (tr) to node[right] {$(-b)-(-d)$} (br);
    \draw[->] (bl) to node[left] {$(-c)-(-a)$} (tl);
    \end{tikzpicture}
    \caption{The structure of Coordination in a zero-sum game. For this to be possible, we would need $c-d > 0$, $b- a>0$, $(-b)-(-d)>0$ and $(-c)-(-a) > 0$, giving an impossible cycle $d > b > a > c > d$.}
    \label{fig:zerosum coordination}
\end{figure}
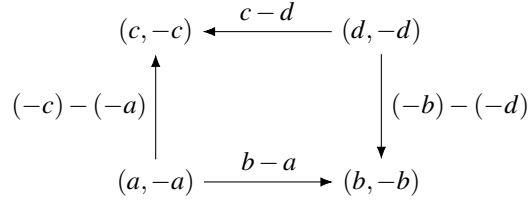
\end{exmp}

This example demonstrates that preference-zero-sum (and strategically-zero-sum) games are a non-trivial set of games. Usefully, the structure of the proof also suggests a way of characterising this set. The critical fact was the existence of a strict \emph{cycle} in the underlying utilities. In some sense, zero-sum games are \emph{acyclic}. To uncover this cycle, we use a transformation we call \emph{reflection}.

\begin{defn} \label{def: reflection}
Let $(u_1,u_2)$ be a two-player game. The \emph{reflected game} is $(u_1,-u_2)$. The \emph{reversed game} is $(-u_1,-u_2)$.
\end{defn}

Note that we made an arbitrary choice here; we could just as easily have defined the reflected game as $(-u_1,u_2)$ (`reflecting' the game in player 1 rather than player 2). These two games are not equivalent, but they are reversals of each other: $-(u_1,-u_2) = (-u_1,u_2)$. 
Our theorems are symmetric under reversal, so both choices work equally well.
Reversing a game has the effect of reversing all arcs in the response graph. 
\begin{defn}[Path-weight]
Let $p = x_1,x_2,\dots,x_n$ be a path in the response graph. The \emph{path-weight} of $p$ is the (signed) sum of arc labels along $p$, that is
\[
\pathweight(p) = \sum_{i=1}^{n-1} (u_{p_i}(x_i) - u_{p_i}(x_{i+1}))
\] where $p_i$ is the unique player such that $x_i$ and $x_{i+1}$ are $p_i$-comparable.
\end{defn}
We get the following theorem:
\begin{thm}[Strategic Zero-Sum--Potential Duality] \label{strategic duality}
A two-player game $(u_1,u_2)$ is strategically-potential if and only if the path-weight of any path between the same two nodes is identical. It is strategically-zero-sum if and only if its reflection $(u_1,-u_2)$ is strategically-potential.
\end{thm}
It follows easily that a potential game cannot have any strict cycles, as any path from a node to itself must have zero path-weight. Recall the analogy with path-independence and conservative vector fields in calculus. Here the path integral is replaced with the sum over weights on a path in the response graph, and one finds that the value of this `path integral' is equal to the difference in potential between the two endpoints of the path. Interestingly, the reflection operation mirrors the relationship between potential and \emph{Hamiltonian} vector fields\cite{alongi2007recurrence, hubbard2015vector}, which are known to be connected to zero-sum games\cite{hofbauer1996evolutionary,sandholm2010population,balduzzi2018mechanics}. Now we find that a combinatorial analogue of this relationship is captured in the response graph. A characterisation of preference-potential and preference-zero-sum games follows easily.
\begin{corol}[Preference Zero-Sum--Potential Duality] \label{corol: preference duality}
A two-player game $(u_1,u_2)$ is preference-potential if and only if every cycle in its response graph contains only undirected edges. It is preference-zero-sum if and only if its reflection $(u_1,-u_2)$ is preference-potential.
\end{corol}
It is clear now that the existence of pure Nash equilibria extends from (generic) potential games to (generic) preference-potential games. As acyclic graphs, every strongly connected component is a singleton, and so all sink components are singletons, and singleton sink components are pure Nash equilibria. It also is immediate that no generic preference-potential game ever contains MP, because this is a cycle. With this theorem in mind, we can return to Example~\ref{coordination not zerosum}. The reflection of Coordination is Matching Pennies---that is, a cycle---and so we conclude immediately that Coordination is not preference-zero-sum. This is shown in Figure~\ref{fig:reflection of coordination}. In fact, given that the reflection of a zero-sum game cannot have any cycles, CO is never contained in any zero-sum game.

\begin{figure}
    \centering
    \begin{tikzpicture}
    \node[nd] (bl) at (0,0) {};
    \node[nd] (br) at (1.5,0) {};
    \node[nd] (tl) at (0,1.5) {};
    \node[nd] (tr) at (1.5,1.5) {};
    \draw[->,blue] (br) to node[above] {} (bl);
    \draw[->,blue] (tl) to node[above] {} (tr);
    \draw[->] (br) to node[right] {} (tr);
    \draw[->] (tl) to node[left] {} (bl);
    
    \node[nd] (bl) at (4,0) {};
    \node[nd] (br) at (5.5,0) {};
    \node[nd] (tl) at (4,1.5) {};
    \node[nd] (tr) at (5.5,1.5) {};
    \draw[->,blue] (bl) to node[above] {} (br);
    \draw[->,blue] (tr) to node[above] {} (tl);
    \draw[->] (br) to node[right] {} (tr);
    \draw[->] (tl) to node[left] {} (bl);
    \end{tikzpicture}
    \caption{The graph of Coordination (left) and its reflection in player 2, which is the graph of Matching Pennies (right). Note that reflecting the game in the preferences of player 1 swaps between the same two graphs.}
    \label{fig:reflection of coordination}
\end{figure}
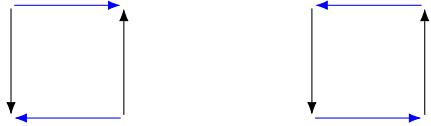

\begin{corol} \label{corol: no CO}
Every weak form of CO contained in a preference-zero-sum game is made up of only undirected edges. Likewise, every weak form of MP contained in a preference-potential game is made up of only undirected edges.
\end{corol}

One of the fundamental results of two-player zero-sum games is that the set of Nash equilibria is convex. Indeed, finding Nash equilibria in a two-player game is equivalent to linear programming\cite{myerson1997game}. In non-degenerate\cite{quint1997theorem} zero-sum games, the Nash is unique, and so there is \emph{at most one} pure Nash equilibria. Surprisingly, this uniqueness generalises to the sink components of preference-zero-sum games.

\begin{defn}[Near-Subgame]
Let $X \subseteq S_1\times S_2$ be a set of pairs. We say $X$ is a \emph{near-subgame} if for each pair $(s_1,s_2)$ and $(t_1,t_2)$ in $X$, at least one of $(s_1,t_2)$ or $(t_1,s_2)$ is in $X$.
\end{defn}

If we required instead that for each $(s_1,s_2)$ and $(t_1,t_2)$ in $X$, \emph{both} $(s_1,t_2)$ or $(t_1,s_2)$ were in $X$, then $X$ is a subgame of the game. The name \emph{near-subgame} reflects the fact that this is a slight weakening of that requirement. In Figure~\ref{fig:outer} we show a game whose sink component is a near-subgame but not also a subgame.

\begin{thm}[Uniqueness of the sink component]\label{one sink}
If a game does not contain Coordination, then the set of sink component profiles is a near-subgame; as a consequence, the game has exactly one sink component.
\end{thm}

Thus we find that CO is responsible for the phenomenon of non-uniqueness of the sink components in games. Consequently it is the cause of the equilibrium selection problem\cite{harsanyi1988general}, at least for pure Nash equilibria. Preference-zero-sum games do not suffer from this problem, because they do not contain CO. Thus simply sharing a response graph with a zero-sum game is sufficient to ensure that there is a unique sink component.

\begin{corol} \label{one PNE}
A preference-zero-sum game has exactly one sink component, and if generic has at most one pure Nash equilibrium.
\end{corol}

\section{The Importance of Matching Pennies and Coordination}

In the previous section we used the $2\times 2$ games Matching Pennies and Coordination as the prototypical examples of preference-zero-sum and preference-potential games respectively. In this section we show that these two games play a key role in introducing strategic complexity to two-player games in any number of strategies.

\begin{thm} \label{ CO and MP theorem}
In any non-dominance-solvable two-player game, every strategy surviving iterated dominance takes part in a subgame that is a weak form of Matching Pennies or Coordination.
\end{thm}
\begin{figure}
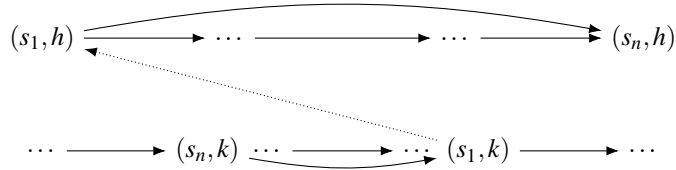

    \centering
    \includestandalone{figs/proof1}
    \caption{\emph{Sketch of Theorem~\ref{ CO and MP theorem}}: Assuming no dominated strategies, we pick a strategy $h$, and label the other players strategies in order $s_1,\dots,s_n$. By assumption, $s_n$ does not dominate $s_1$, so we can find other another strategy $k$ where $s_1$ and $s_n$ are reversed. Pick a direction for the arc from $(s_1,h)$ to $(s_1,k)$ (the dotted arc). Requiring there be no MP or CO subgames forces all remaining arcs from $(s_i,k)$ to $(s_i,h)$ to be in the same direction; we obtain a contradiction where either $h$ dominates $k$ or $k$ dominates $h$.}
    \label{fig:dominance setup}
\end{figure}
The proof is given in full the Supplementary Material. The key ideas are given in Figure~\ref{fig:dominance setup}.

A consequence of this theorem is that \emph{two-player games inherit dominance-solvability from their $2\times 2$ subgames}. That is, if every $2\times 2$ subgame has a dominated strategy, then the game is dominance-solvable.
Being dominance-solvable, games without MP or CO are somewhat trivial, and so MP and CO are responsible for bringing strategic complexity to a game. In a similar way, we established above that CO brings the problem of equilibrium selection to a game. Recall Corollary~\ref{corol: no CO}: preference-zero-sum games do not contain CO, and preference-potential games do not contain MP. This theorem immediately gives us a partial converse: generic preference-zero-sum (respectively preference-potential) games either contain MP (respectively CO), or are dominance-solvable.
\begin{corol} \label{corol: zero sum MP}
Every strategy in a non-dominance-solvable preference-zero-sum game takes part in an MP subgame. Likewise, every strategy in a non-dominance-solvable preference-potential game takes part in an CO subgame.
\end{corol}
Matching Pennies is truly the prototypical preference-zero-sum response graph---not only is it the simplest example of such, but all non-dominance-solvable preference-zero-sum games contain it. The same is true for Coordination and two-player preference-potential games. As a consequence, we find that the intersection of games which are both generic preference-zero-sum and generic preference-potential can contain neither MP nor CO, and thus must be dominance-solvable. This ties together our characterisations of preference-potential and preference-zero-sum games and connects them to the concept of iterated dominance.

\begin{corol}[Zero-Sum--Potential--Dominance Theorem] \label{dominance theorem}
Any generic game that is both preference-zero-sum and preference-potential is dominance-solvable.
\end{corol}

To demonstrate this result, consider the generic $2\times 2$ games. As acyclic graphs, DD, SD and CO are preference-potential. As reflected acyclic graphs, DD, SD and MP are preference-zero-sum (Corollary~\ref{corol: preference duality}). As both preference-potential and preference-zero-sum games, DD and SD are dominance-solvable. Thus \emph{every} generic $2\times 2$ game is either preference-potential, preference-zero-sum, or dominance-solvable.

This highlights another important point: the sets of preference-zero-sum games and preference-potential games are quite broad, much more so than zero-sum or even strategically zero-sum games. In a similar result\cite{hwang2020strategic}, the authors proved that any two-player game both strategically-potential and strategically-zero-sum must have a \emph{dominant strategy} for each player. While this is an interesting result, its scope is more limited than Corollary~\ref{dominance theorem}; any game preference-equivalent to both a zero-sum and potential game is certainly strategically equivalent to both, but the converse does not hold. For instance, no game with the response graph of SD can ever be strategically-potential and strategically-zero-sum (there are no weights such that the graph and its reflection are both have the same path-weights on all paths between the same profiles) but SD is preference-zero-sum and preference-potential and so falls under the wider purview of our theorem. There are even cases\cite{li2020verification} of the explicit study of the yet-more-restricted case of games that are both zero-sum and potential, without it being noted that these games are all dominance-solvable.

\section{Applications} \label{sec:applications}

The generic $2\times 2$ games (Figure~\ref{fig: 2x2 response graphs}) have served as useful examples throughout this paper. This is particularly true of MP and CO, the two without dominated strategies, which also served as our prototypical examples of preference-zero-sum and preference-potential games. However, not all interesting properties of two-player games can be captured in just these graphs. In this section we discuss how the properties of response graph, particularly being preference-zero-sum and preference-potential, extend to larger two-player games, such as $2\times 3$, $2\times 4$ and $3\times 3$. It is our goal to build the reader's intuition about response graphs and to provide a stock of example graphs with interesting game-theoretic properties. We also intend to demonstrate how the theorems of the paper can help us to analyse games. To keep things simple, we will focus on games without dominated strategies.
Games which possess dominated strategies can be simplified into smaller games by deleting those strategies\cite{myerson1997game}. By Corollary~\ref{dominance theorem}, these graphs split into three categories: preference-zero-sum only, preference-potential only, and neither---games which are both preference-zero-sum and preference-potential always have a dominated strategy.

We begin with the generic $2\times 3$ games. There are exactly three such graphs without dominated strategies, as the following argument shows: in order for there to be no dominated strategies, the three-strategy player (we assume player 1) must prefer their strategies in opposite orders for each of the two strategies of player 2. It remains only to choose player 2's preferences for each choice of strategy for player 1, which leads to the three graphs shown in Figure~\ref{fig: 2x3}. We can distinguish these graphs via their $2\times 2$ subgames: Figure~\ref{fig:2x3 MP} contains two MP subgames and a SD subgame, so is preference-zero-sum; Figure~\ref{fig:2x3 CO} contains two CO subgames and a SD subgame, so is preference-potential, and in fact is the reflection of \ref{fig:2x3 MP}; Figure~\ref{fig:2x3 MPCO} contains MP, CO and SD subgames and so is neither preference-potential nor preference-zero-sum, and so is the \emph{unique minimal example} of such a graph. We call these `$2\times 3$ MP', `$2\times 3$ CO' and `$2\times 3$ MP-CO' by analogy with the $2\times 2$ case. All of these graphs are isomorphic to themselves under reversal, and the $2\times 3$ MP-CO game is also isomorphic to itself after reflection of either player. A similar argument also works to classify the $2\times 4$ games. In that case there are 9 distinct response graphs with no dominated strategies (Figure~\ref{fig: 2x4}): two are preference-zero-sum (\ref{fig:0001} and \ref{fig:0011}), two preference-potential (\ref{fig:1100} and \ref{fig:1000}), and five neither.

There are 156 distinct response graphs of $3\times 3$ generic games without dominated strategies, which can be found by a computer search. Of these, 25 are preference-zero-sum and 30 are preference-potential, and the remaining 101 are neither (Corollary~\ref{dominance theorem}). We will now discuss a few of these which are useful examples of particular game-theoretic properties. In generic $2\times 2$ and $2\times 3$ games, any game without an MP was acyclic and thus preference-potential. In generic $3\times 3$ this becomes no longer true, and there are exactly two graphs (Figure~\ref{fig:upper} and~\ref{fig:lower}) which do not contain a 4-cycle (an MP) and yet do contain a 6-cycle (it is easy to see that every $3\times 3$ game which has a 5-cycle must have a 4-cycle). One way to see this is by applying Theorem~\ref{ CO and MP theorem}, using the following argument: suppose our $3\times 3$ game has six profiles which take part in a cycle, and no smaller cycle. By Theorem~\ref{ CO and MP theorem}, each strategy must participate in a CO subgame, as we have assumed there are no MP subgames. One finds only two possibilities: either the three remaining nodes are all sources with arcs into the 6-cycle (Figure~\ref{fig:upper}), or they are all sinks with arcs from the 6-cycle (Figure~\ref{fig:lower}). We call these graphs the 6-cycle-source and -sink graphs respectively. By the same reasoning, the reflection of either of these graphs in either player must give a graph which has no CO subgame and yet is not preference-zero-sum. All choices of reflection are in fact isomorphic, giving a graph which we call the Reflected 6-cycle graph (Figure~\ref{fig:reflected}). This is the smallest example of a generic $3\times 3$ game without dominated strategies which does not contain CO but is also not preference-zero-sum. Its reversal is itself.

Corollary~\ref{corol: zero sum MP} and Corollary~\ref{dominance theorem} together tell us that in preference-zero-sum games without dominated strategies, every strategy must take part in a MP subgame. This leads to response graphs which are typically highly connected, like MP itself. However, in $3\times 3$ games we find the first examples of preference-zero-sum games which are \emph{not} strongly connected and yet have no dominated strategies. We call these the Inner and Outer Diamond graphs, for their shape (Figure~\ref{fig:inner} and~\ref{fig:outer}). They are reversals of each other. These two graphs are useful examples: the Inner Diamond graph is the smallest game demonstrating that generic preference-zero-sum games can have pure Nash equilibria. The Outer Diamond game is the smallest example of a generic preference-zero-sum game whose sink component is not a subgame (though it is a \emph{near-subgame}, by Theorem~\ref{one sink}).

\begin{figure}
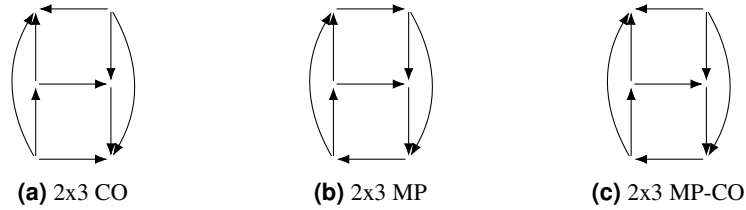

    \centering
    \begin{subfigure}{0.22\textwidth}
    \centering
    \includestandalone{figs/2x3/CO}
    \caption{2x3 CO}
    \label{fig:2x3 CO}
    \end{subfigure}
    \begin{subfigure}{0.22\textwidth}
    \centering
    \includestandalone{figs/2x3/MP}
    \caption{2x3 MP}
    \label{fig:2x3 MP}
    \end{subfigure}
    \begin{subfigure}{0.22\textwidth}
    \centering
    \includestandalone{figs/2x3/MPCO}
    \caption{2x3 MP-CO}
    \label{fig:2x3 MPCO}
    \end{subfigure}
    \caption{The three response graphs of generic non-dominated $2\times 3$ games}
    \label{fig: 2x3}
\end{figure}

\begin{figure}
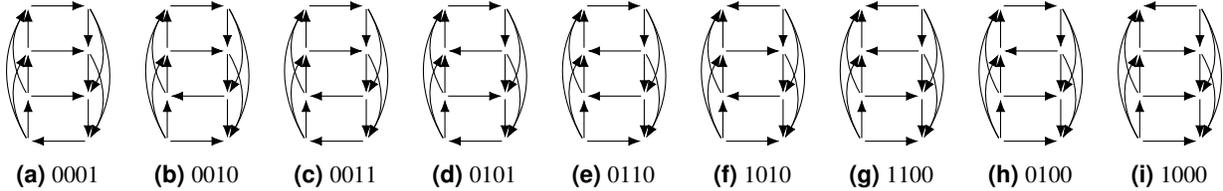

    \centering
    \begin{subfigure}{0.1\textwidth}
    \centering
    \includestandalone{figs/2x4/1-0001}
    \caption{0001}
    \label{fig:0001}
    \end{subfigure}
    \begin{subfigure}{0.1\textwidth}
    \centering
    \includestandalone{figs/2x4/2-0010}
    \caption{0010}
    \label{fig:0010}
    \end{subfigure}
    \begin{subfigure}{0.1\textwidth}
    \centering
    \includestandalone{figs/2x4/3-0011}
    \caption{0011}
    \label{fig:0011}
    \end{subfigure}
    \begin{subfigure}{0.1\textwidth}
    \centering
    \includestandalone{figs/2x4/4-0101}
    \caption{0101}
    \label{fig:0101}
    \end{subfigure}
    \begin{subfigure}{0.1\textwidth}
    \centering
    \includestandalone{figs/2x4/5-0110}
    \caption{0110}
    \label{fig:0110}
    \end{subfigure}
    \begin{subfigure}{0.1\textwidth}
    \centering
    \includestandalone{figs/2x4/6-1010}
    \caption{1010}
    \label{fig:1010}
    \end{subfigure}
    \begin{subfigure}{0.1\textwidth}
    \centering
    \includestandalone{figs/2x4/7-1100}
    \caption{1100}
    \label{fig:1100}
    \end{subfigure}
    \begin{subfigure}{0.1\textwidth}
    \centering
    \includestandalone{figs/2x4/8-0100}
    \caption{0100}
    \label{fig:0100}
    \end{subfigure}
    \begin{subfigure}{0.1\textwidth}
    \centering
    \includestandalone{figs/2x4/9-1000}
    \caption{1000}
    \label{fig:1000}
    \end{subfigure}
    \caption{The nine response graphs of generic non-dominated $2\times 4$ games. The binary encoding defines the direction of each arc for player 2. \ref{fig:0001} and \ref{fig:0011} are preference-zero-sum, and \ref{fig:1100} and \ref{fig:1000} are preference-potential.}
    \label{fig: 2x4}
\end{figure}

\begin{figure}
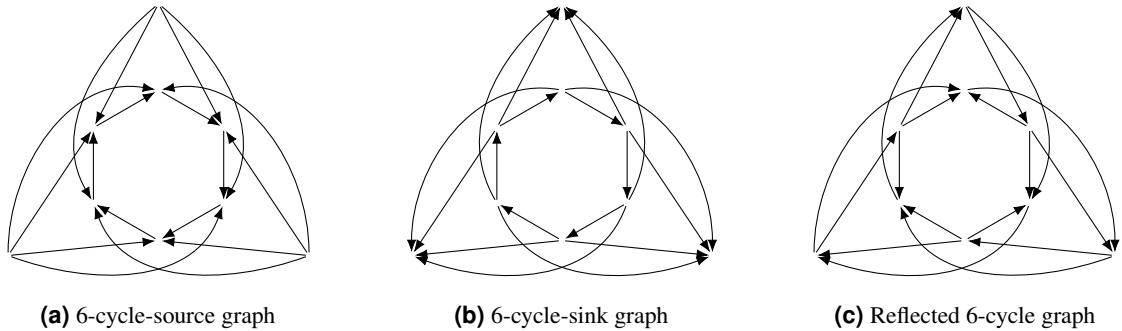

    \centering
    \begin{subfigure}{0.3\textwidth}
    \centering
    \includestandalone{figs/3x3/upper_sixloop}
    \caption{6-cycle-source graph}
    \label{fig:upper}
    \end{subfigure}
    \begin{subfigure}{0.3\textwidth}
    \centering
    \includestandalone{figs/3x3/lower_sixloop}
    \caption{6-cycle-sink graph}
    \label{fig:lower}
    \end{subfigure}
    \begin{subfigure}{0.3\textwidth}
    \centering
    \includestandalone{figs/3x3/reflected_sixloop}
    \caption{Reflected 6-cycle graph}
    \label{fig:reflected}
    \end{subfigure}
    \caption{The 6-cycle-source graph (\ref{fig:upper}) and its reversal, the 6-cycle-sink graph (\ref{fig:lower}), do not contain MP and yet are not preference-potential. The reflection of either game in either player gives the Reflected 6-cycle graph (\ref{fig:reflected}), which is the unique $3\times 3$ game which does not contain CO yet is not preference-zero-sum.}
    \label{fig: sixloop}
\end{figure}

\begin{figure}
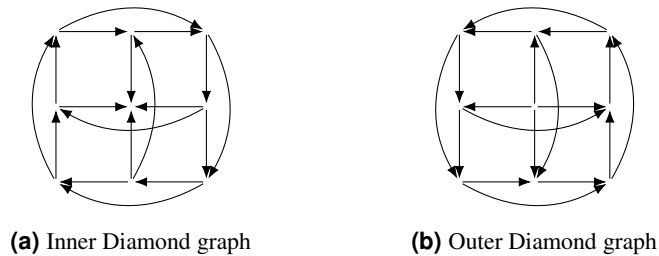

    \centering
    \begin{subfigure}{0.3\textwidth}
    \centering
    \includestandalone{figs/3x3/inner_diamond}
    \caption{Inner Diamond graph}
    \label{fig:inner}
    \end{subfigure}
    \begin{subfigure}{0.3\textwidth}
    \centering
    \includestandalone{figs/3x3/outer_diamond}
    \caption{Outer Diamond graph}
    \label{fig:outer}
    \end{subfigure}
    \caption{The Inner and Outer Diamond graphs}
    \label{fig: diamond}
\end{figure}

\section{Conclusions}

In this paper we discussed the response graphs of two-player games. The response graph is a model of game which captures only the underlying notion of strategic preference and not the cardinal values for payoffs, in other words, a model that does not concern itself with the actual payoff values and only focuses on the ordering of the payoffs. This allows response graphs to be used as a model in circumstances when access to or knowledge of cardinal payoffs is implausible. The notion of \emph{preferences} agrees with our intuitive notions about simple games such as Rock-Paper-Scissors. While many key game-theoretic concepts---such as dominance and pure Nash equilibria---depend only on which strategies are preferred, the response graph has received little direct study and few of its general properties are known. In this paper we demonstrated that the response graph contains significant mathematical structure, and its study leads to new game-theoretic insight. We showed first that two-player potential and zero-sum games, two of the best-studied classes of game, have very natural characterisations in terms of response graph structure---specifically, acyclicity. Furthermore, we established that the key equilibrium properties of these games translate to analogous properties of the sink component of the response graph. We then argued that the response graphs of the Matching Pennies and Coordination games play a key role in two-player games: any game not including these games as subgames must be dominance-solvable. In summary, we found, and strove to convey the message, that the response graph is an interesting game-theoretic object which we believe merits further study.


\bibliography{references}


\appendix

\section{Proofs}

\begin{thm}[Theorem~\ref{reconstruction}]
Given a graph $G$, we can construct a game whose response graph is $G$, or determine that no such games exist, in time linear in the number of arcs.
\end{thm}
\begin{proof}
In \cite[Theorem 22.2]{hammack2011handbook} it is established that Hamming graphs, the underlying graphs of response graphs, can be recognised and given a labelling by tuples in $S_1\times S_2\times \dots\times S_N$ such that two adjacent nodes differ in a single entry of their tuples. This can be done in time $O(m)$ where $m$ is the number of edges. This labelling is unique up to the choice of sets $S_i$.

Let $G$ be a given graph. To check if $G$ is the response graph of a game, we first check if its underlying graph is a Hamming graph, using the above technique. If so, its node set is $S_1\times S_2\times \dots\times S_n$, and we assign these sets as the strategy sets for each player. It remains only to check that the graph is oriented such that, for each fixed choice of strategies for $N-1$ players, the strategy profiles for each choice of the remaining player are totally ordered. We can do this by iteration; for each player $i$, iterate through each combination of strategies for the other players and verify that the associated subgraph is directed such that it is acyclic. If not, we reject the graph. This loop examines each edge of the graph once, so is $O(m)$.

Now we show that any graph $G$ satisfying these criteria is indeed a response graph. Fix some player $i$, with $|S_i| = k$. For each fixed choice of strategy $s_{-i}$ to all players other than $i$, the associated subgraph of the response graph is a total order. We assign the payoffs $1,2,\dots,k$ to the strategies in $S_i$, in this order. The result is a game whose response graph is $G$.
\end{proof}

\begin{lem}[Lemma~\ref{lem: subgame lemma}]
If the response graph of a two-player game contains the response graph of a $2\times 2$ game, then the profiles which take part form a $2\times 2$ subgame.
\end{lem}
\begin{proof}
The result is a special case of a general property of Hamming graphs.

Suppose $p_1,p_2,p_3,p_4$ are profiles, and the underlying graph of the subgame induced by them is a 4-cycle, with nodes in this order. Let $p_1 = (a_1,b_1)$ and $p_2 = (a_2,b_1)$ without loss of generality. Then $p_3$ is comparable to $p_2$ but not $p_1$, so again without loss of generality $p_3 = (a_2,b_2)$. Finally, $p_4$ is comparable to $p_1$ and $p_3$ but not $p_2$, and so we conclude that $p_4 = (a_1,b_2)$. Thus these profiles correspond to the subgame $\{a_1,a_2\}\times \{b_1,b_2\}$. 
\end{proof}

\begin{thm}[Theorem~\ref{strategic duality}]
A two-player game $(u_1,u_2)$ is strategically-potential if and only if the path-weight of any path between the same two nodes is identical. It is strategically-zero-sum if and only if its reflection $(u_1,-u_2)$ is strategically-potential.
\end{thm}
\begin{proof}
\emph{Claim}: A game $(u_1,u_2)$ is strategically-potential if and only if all undirected paths between any two profiles $x_1$ and $x_n$ have the same path-weight.

Suppose $(u_1,u_2)$ is potential with potential function $\phi$, and let $p_1 = x_1,x_2,\dots,x_n$ and $p_2 = x_1,y_1,\dots,y_m,x_n$ be two paths between profiles $x_1$ and $x_2$. The path-weight is
\begin{align*}
    \pathweight(p_1) &= \sum_{i=1}^{n-1} (u_{p_i}(x_i) - u_{p_i}(x_{i+1}))\\ 
    \pathweight(p_1) &= \sum_{i=1}^{n-1} (\phi(x_i) - \phi(x_{i+1})) \qquad \text{(the game is potential)}\\
    \pathweight(p_1) &= \phi(x_1) - \phi(x_n)
\end{align*}
where we have used the fact that the sum is telescoping. By identical reasoning, $\pathweight(p_2) = \phi(x_1) - \phi(x_2) = \pathweight(p_1)$.

For the converse, suppose that in $(u_1,u_2)$ the path-weight on any path between the same two nodes is equal. Define an order $v\preceq w$ if the path-weight of any path from $v$ to $w$ is non-negative. This is well-defined because all such paths have the same path-weight. This is reflexive, and we can see that it is transitive by the following. If $v\preceq w\preceq t$, then the path-weight from $v$ to $t$ is the sum of the path-weights of paths from $v$ to $w$ and $w$ to $t$ respectively, and these are each non-negative, so the path-weight from $v$ to $t$ is also non-negative.Since the underlying graph is connected, this order is \emph{total}. This order must have minimal elements as it is finite. Choose one, call it $z$. Define a potential function $\phi$ as follows. Set $\phi(z) = 0$, and for any other node $x$, define $\phi(x) = \pathweight(p_{x\to z})$, where $p_{x\to z}$ is any undirected path from $x$ to $z$.

Now we show this is indeed a potential function. Let $v$ and $w$ be $i$-comparable profiles. Choose paths $p_{v\to z}$ and $p_{z\to w}$. Since $u$ is path-independent, the path-weight along the one-step path $p = u,v$, which is $u_i(v) - u_i(w)$, must be equal to the path-weight of the concatenated path $p_{u\to z}p_{z\to v}$, and this is equal to $\phi(u) + (-\phi(v))$, so $\phi$ is a potential function. This establishes the claim.

Observe that $(u_1,u_2)$ and $(v_1,v_2)$ are strategically equivalent if and only if $(u_1,-u_2)$ and $(v_1,-v_2)$ are strategically equivalent. Suppose $(u_1,u_2)$ is zero-sum. Then the payoff in any profile $(s_i,s_j)$ is $(x_{i,j},-x_{i,j})$ for some real $x_{i,j}$. In the reflected game $(u_1,-u_2)$ the payoff is $(x_{i,j},x_{i,j})$. This game is an \emph{identical interest game}, and thus a potential game, with potential function $\phi : S_1\times S_2 \to \real$, $\phi(s_i,s_j) = x_{i,j}$. Thus if $(v_1,v_2)$ is strategically-equivalent to $(u_1,u_2)$, then $(v_1,-v_2)$ is strategically-equivalent to $(u_1,-u_2)$, so is strategically-potential. For the converse, suppose that $(u_1,u_2)$ is potential, with potential function $\phi :S_1\times S_2 \to\real$. Then the game $(\phi,-\phi)$ is clearly a zero-sum game, and its reflection $(\phi,\phi)$ is a potential game with potential $\phi$ by the above. For either player $i$ and $i$-comparable profiles $v$ and $w$, $u_i(v) - u_i(w) = \phi(v) - \phi(w)$ and so $(u_1,u_2)$ is strategically equivalent to $(\phi,\phi)$. By transitivity, any game strategically equivalent to $(u_1,u_2)$ has a reflection which is strategically-zero-sum, as it is strategically equivalent to $(\phi,-\phi)$.
\end{proof}

\begin{corol}[Corollary~\ref{corol: preference duality}]
A two-player game $(u_1,u_2)$ is preference-potential if and only if every cycle in its response graph contains only undirected edges. It is preference-zero-sum if and only if its reflection $(u_1,-u_2)$ is preference-potential.
\end{corol}
\begin{proof}
If the game is potential with potential function $\phi$, then for any cycle $x_1,\dots,x_n,x_1$ we have $\phi(x_1) \leq \phi(x_2) \leq \dots \leq \phi(x_n) \leq \phi(x_1)$, which implies $\phi(x_1) = \phi(x_2) = \dots = \phi(x_n)$. Thus the cycle contains only undirected edges. 

For the converse, we will use an approach similar to Theorem~\ref{strategic duality}, where we will construct a potential function and argue that the associated potential game has this response graph. Suppose that in $(u_1,u_2)$ every cycle contains only undirected edges. Define an order $v\preceq w$ if there is a directed path from $v$ to $w$. This is the reachability partial order of the graph. We can always assign real numbers $r_i$ to each strongly connected component $s_i$ such that $r_i < r_j$ if $s_i \prec s_j$. Choose some such values, and then define $\phi(v) = r_i$, where $v$ is in the connected component $s_i$.

To see that this is a potential function with the same response graph, let $v$ and $w$ be $i$-comparable profiles. If the arc between $v$ and $w$ is an undirected edge, then $v\preceq w$ and $w\preceq v$ so they are in the same connected component and $\phi(v) = \phi(w)$. Conversely, if $\phi(v) = \phi(w)$ then $v\preceq w$ and $w\preceq v$, and so there is a cycle from $v$ to itself containing $w$; by the above, all arcs on this cycle are undirected edges, so the arc between $v$ and $w$ is an undirected edge. If the arc is not an undirected edge, so $\arc{v}{w}$, then $v\prec w$ and by construction $\phi(v) < \phi(w)$. Conversely, if $\phi(v) < \phi(w)$ then $v$ and $w$ are in different connected components, and the component of $v$ precedes that of $w$ in the reachability order, so there is an arc $\arc{v}{w}$. Hence $\phi$ defines a potential game with the same response graph as $(u_1,u_2)$.

Now suppose $(u_1,u_2)$ is preference-zero-sum, so preference-equivalent to some zero-sum game $(z_1,z_2)$. By Theorem~\ref{strategic duality}, $(z_1,-z_2)$ is potential, and $(u_1,-u_2)$ is preference-equivalent to $(z_1,-z_2)$. For the converse, suppose $(u_1,-u_2)$ is preference-equivalent to some potential game $(w_1,w_2)$. By Theorem~\ref{strategic duality}, $(w_1,-w_2)$ is strategically equivalent to a zero-sum game $(z_1,z_2)$. By transitivity, and the fact that strategically equivalent games are preference-equivalent, $(u_1,u_2)$ is preference-equivalent to $(z_1,z_2)$.
\end{proof}

\begin{corol}[Corollary~\ref{corol: no CO}]
Every weak form of CO contained in a preference-zero-sum game is made up of only undirected edges. Likewise, every weak form of MP contained in a preference-potential game is made up of only undirected edges.
\end{corol}
\begin{proof}
If a game contains a weak form of CO then it is a subgame (Lemma~\ref{lem: subgame lemma}). Any such subgame becomes a subgame isomorphic to a weak form of MP in the reflected game. However, in a preference-potential game, any MP subgame must consist of only undirected edges, and so in a preference-zero-sum game any CO must consist of only undirected edges, by Theorem~\ref{strategic duality}.
\end{proof}

\begin{thm}[Theorem~\ref{one sink}]
If a game does not contain Coordination, then the set of sink component profiles is a near-subgame; as a consequence, the game has exactly one sink component.
\end{thm}
\begin{proof}
Let $(a,b)$ and $(x,y)$ be profiles contained in (possibly different) sink components. For contradiction, assume that neither $(a,y)$ nor $(x,b)$ are in sink components, and so neither $(a,b)$ nor $(x,y)$ have arcs to them. However, we find that the subgame $\{a,x\}\times \{b,y\}$ is the response graph of CO, which contradicts our assumption. Thus the sink component profiles are a near-subgame, and without loss of generality there is an arc $\arc{(a,b)}{(x,b)}$. Now we show uniqueness by demonstrating that $(a,b)$ and $(x,y)$ must be in the same sink component. Strongly connected components define an equivalence relation on nodes, so we need only show that there is a node in common. Since there is an arc $\arc{(a,b)}{(x,b)}$, $(x,b)$ is in the same sink component as $(a,b)$. If there is an arc $\arc{(x,b)}{(x,y)}$ then $(x,y)$ is also in this component; if there is an arc $\arc{(x,y)}{(x,b)}$ then $(x,y)$ is in the same sink component as $(x,b)$ and so, by transitivity, $(a,b)$. As at least one of these arcs must exist, the result is proved.
\end{proof}

\begin{corol}[Corollary~\ref{one PNE}]
A preference-zero-sum game has exactly one sink component, and if generic has at most one pure NE.
\end{corol}
\begin{proof}
Preference-zero-sum games do not contain CO (Corollary~\ref{corol: no CO}), and so there is one sink component by Theorem~\ref{one sink}. Pure Nash equilibrium are singleton sink components. As there is exactly one sink component, this is the only pure NE.
\end{proof}

\begin{thm}[Theorem~\ref{ CO and MP theorem}]
In any non-dominance-solvable two-player game, every strategy surviving iterated dominance takes part in a subgame that is a weak form of Matching Pennies or Coordination.
\end{thm}
\begin{proof}
We assume for contradiction that all iteratively dominated strategies have been removed. As the game is not dominance-solvable, there are at least two strategies remaining for both players and no strategy dominates another. Now let $h$ be some strategy for player 2 (without loss of generality), and suppose that $h$ never takes part in a subgame isomorphic to a weak form of MP or CO. We will demonstrate a contradiction by showing that a dominated strategy must exist. Let $s_1,\dots,s_n$ be the strategies for player 1, ordered by payoff for player 1 when player 2 plays $h$, that is, $u_1(s_1,h) \leq u_1(s_2,h) \leq \dots \leq u_1(s_n,h)$. This is pictorially shown in Figure~\ref{fig:dominance setup}.

Suppose first that $u_1(s_1,h) = u_1(s_n,h)$, then $u_1(s_i,h) = u_1(s_j,h)$ for any $s_i$ and $s_j$. Pick any other strategy $k$ for player 2. If there exists $s_i$ and $s_j$ where $u_2(s_i,h) \leq u_2(s_i,k)$ and $u_2(s_j,h) \geq u_2(s_j,k)$, then $\{s_i,s_j\}\times \{h,k\}$ is a weak form of MP or CO. If there do not exist such a pair of $s_i$ and $s_j$, then we find that $u_2(s_i,h) < u_2(s_i,k)$ for every $s_i$ or $u_2(s_i,h) > u_2(s_i,k)$ for every $s_i$, implying that $h$ dominates or is dominated by $k$, contradicting our assumption.

Now suppose that $u_1(s_1,h) < u_1(s_n,h)$. No strategy dominates any other, so $s_n$ does not dominate $s_1$, and thus there is some strategy $k$ for player 2 where $u_1(s_n,k) \leq u_1(s_1,k)$. If $u_1(s_1,h) = u_1(s_1,k)$, then the subgame $\{s_1,s_n\}\times \{h,k\}$ is a weak form of MP or CO. Assuming this does not hold, there are two cases, for which the argument is symmetric while only reversing the role of $h$ and $k$ and MP and CO. In case (1), player 2 prefers $h$ when player 1 plays $s_1$, so $u_2(s_1,h) > u_2(s_1,k)$, and case (2) is the opposite, where $u_2(s_1,h) < u_2(s_1,k)$.  Case (1) is depicted in Figure~\ref{fig:dominance setup}.

Now let $s_i$ be any strategy with $u_1(s_i,k) \leq u_1(s_1,k)$. Since $s_1$ is least preferred by player 1 when player 2 plays $h$, we also have $u_1(s_1,h) \leq u_1(s_i,h)$. By case (1) we also have $u_2(s_1,k) < u(s_i,h)$. Thus we cannot also have $u_2(s_i,h) \leq u_2(s_i,k)$, as in that case the subgame $\{s_1,s_i\}\times \{h,k\}$ would be a weak form of MP. Thus for any such $s_i$, $u_2(s_i,k) < u_2(s_i,h)$. Figure~\ref{proof:2} visually summarises this argument. In particular, this implies that $u_2(s_n,k) < u_2(s_n,h)$.

Now let $s_j$ be any strategy with $u_1(s_n,k) \leq u_1(s_j,k)$. Since $s_n$ is most preferred by player 1 when player 2 plays $h$, we also have $u_1(s_j,h) \leq u_1(s_n,h)$. By the above, we also have $u_2(s_n,k) < u_2(s_n,h)$. Thus we cannot also have $u_2(s_j,h) \leq u_2(s_j,k)$, as $\{s_n,s_j\}\times \{h,k\}$ would be a weak form of CO. Thus for any such $s_j$, $u_2(s_j,k) > u_2(s_j,h)$.  Figure~\ref{proof:3} visually summarises this argument.

However, we have now discussed all strategies $s_i\in S_1$, and in each case $u_2(s_i,k) < u_2(s_i,h)$, and so we find that strategy $h$ dominates strategy $k$ for player 2, contradicting our original assumption. Case (2) follows identical reasoning, swapping CO and MP and concluding with $k$ dominating $h$.
\end{proof}

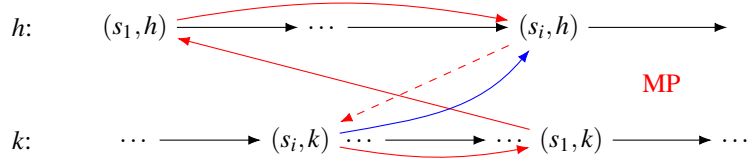
\begin{figure}
        \centering
    \begin{tikzpicture}
    
    \node (h) at (-1.5,1.5) {$h$:};
    
    \node (end) at (8,1.5) {};
    \node (s1) at (0,1.5) {$(s_1,h)$};
    \node (d1) at (2.5,1.5) {$\dots$};
    \node (sn) at (5.5,1.5) {$(s_i,h)$};
    \draw[->] (s1) to (d1);
    \draw[->] (d1) to (sn);
    \draw[->] (sn) to (end);
    \draw[->,red] (s1) to [out = 10, in = 170] (sn);
    
    \node (k) at (-1.5,0) {$k$:};
    \node (s1b) at (0,0) {$\dots$};
    \node (swb) at (2.2,0) {$(s_i,k)$};
    \node (d2b) at (3,0) {$\dots$};
    \node (d3b) at (5,0) {$\dots$};
    
    \node (sab) at (5.8,0) {$(s_1,k)$};
    
    \node (snb) at (8,0) {$\dots$};
    \draw[->] (s1b) to (swb);
    \draw[->] (d2b) to (d3b);
    \draw[->] (sab) to (snb);
    \draw[->,red] (swb) to [out = -10, in = -170] (sab);
    
    \draw[->,red] (sab) to (s1);
    
    \draw[->,red,dashed] (sn) to (swb);
    \draw[->,blue] (swb) to [in = -130, out = 10] (sn);
    \node at (7,0.75) {\color{red}{MP}};
    
 \end{tikzpicture}
 \caption{For $s_i$ where $(s_i,k)$ precedes $(s_1,k)$, the arc $\arc{(s_i,k)}{(s_i,h)}$ must go from $k$ to $h$ (blue); if instead the arc goes from $h$ to $k$ (red, dashed) the subgame $\{s_i,s_1\}\times \{h,k\}$ is MP.}
 \label{proof:2}
\end{figure}
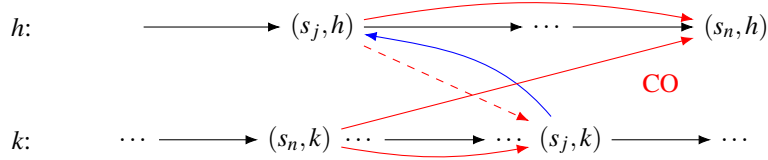
\begin{figure}
        \centering
    \begin{tikzpicture}
    
    \node (h) at (-1.5,1.5) {$h$:};

    \node (start) at (0,1.5) {};
    \node (s1) at (2.5,1.5) {$(s_j,h)$};
    \node (d2) at (5.5,1.5) {$\dots$};
    \node (sn) at (8,1.5) {$(s_n,h)$};
    \draw[->] (start) to (s1);
    \draw[->] (s1) to (d2);
    \draw[->] (d2) to (sn);
    \draw[->,red] (s1) to [out = 10, in = 170] (sn);
    
    \node (k) at (-1.5,0) {$k$:};
    \node (s1b) at (0,0) {$\dots$};
    \node (swb) at (2.2,0) {$(s_n,k)$};
    \node (d2b) at (3,0) {$\dots$};
    \node (d3b) at (5,0) {$\dots$};
    
    \node (sab) at (5.8,0) {$(s_j,k)$};
    
    \node (snb) at (8,0) {$\dots$};
    \draw[->] (s1b) to (swb);
    \draw[->] (d2b) to (d3b);
    \draw[->] (sab) to (snb);
    \draw[->,red] (swb) to [out = -10, in = -170] (sab);
    
    \draw[->,red,dashed] (s1) to (sab);
    \draw[->,blue] (sab) to [out = 130, in = -10] (s1);
    
    \draw[->,red] (swb) to (sn);
    \node at (7,0.75) {\color{red}{CO}};
    
 \end{tikzpicture}
 \caption{For $s_j$ where $(s_j,k)$ succeeds $(s_n,k)$, the arc $\arc{(s_j,k)}{s_j,h)}$ must go from $k$ to $h$ (blue); if instead the arc goes from $h$ to $k$ (red, dashed) the subgame $\{s_j,s_n\}\times \{h,k\}$ is CO.}
 \label{proof:3}
\end{figure}

\begin{corol}[Corollary~\ref{corol: zero sum MP}]
Every strategy in a non-dominance-solvable preference-zero-sum game takes part in an MP subgame. Likewise, every strategy in a non-dominance-solvable preference-potential game takes part in an CO subgame.
\end{corol}
\begin{proof}
By Theorem~\ref{ CO and MP theorem}, every strategy surviving iterated dominance in a game takes part in a subgame which is a weak form of CO or MP. If the game is preference-zero-sum, then every CO subgame is also a weak form of MP, and likewise in a preference-potential game every MP subgame is also a weak form of CO, proving the result.
\end{proof}

\begin{corol}[Corollary~\ref{dominance theorem}]
Any generic game that is both preference-zero-sum and preference-potential is dominance-solvable.
\end{corol}
\begin{proof}
By Corollary~\ref{corol: no CO}, generic preference-zero-sum games do not contain CO, and generic preference-potential games do not contain MP, but games containing neither are dominance-solvable by Theorem~\ref{dominance theorem}.
\end{proof}

\end{document}

%% file: figs/3x3/RPS_labeled.tex
\begin{tikzpicture}

    \node[nd] (1) at (0,0) {$(S,R)$};
    \node[nd] (2) at (2,0) {$(S,P)$};
    \node[nd] (3) at (4,0) {$(S,S)$};
    \node[nd] (4) at (0,1) {$(P,R)$};
    \node[nd] (5) at (2,1) {$(P,P)$};
    \node[nd] (6) at (4,1) {$(P,S)$};
    \node[nd] (7) at (0,2) {$(R,R)$};
    \node[nd] (8) at (2,2) {$(R,P)$};
    \node[nd] (9) at (4,2) {$(R,S)$};
    
    \draw[->] (1) to (4);
    \draw[->] (3) to[in = -30, out = 210] (1);
    \draw[->] (3) to[out = 40, in = -40] (9);
    \draw[->] (9) to (8);
    \draw[->] (8) to (5);
    \draw[->] (4) to (5);
    
    \draw[->] (7) to (4);
    \draw[->] (1) to[out = 140, in = -140] (7);
    \draw[->] (7) to (8);
    \draw[->] (9) to[out = 150, in = 30] (7);
    \draw[->] (2) to (1);
    \draw[->] (2) to (3);
    \draw[->] (5) to (2);
    \draw[->] (8) to[out = -40, in = 40] (2);
    \draw[->] (6) to (3);
    \draw[->] (6) to (9);
    \draw[->] (5) to (6);
    \draw[->] (4) to[in = 200, out = -20] (6);

 \end{tikzpicture}

%% file: figs/3x3/RPS.tex
\begin{tikzpicture}

    \node[nd] (1) at (0,-0.6) {};
    \node[nd] (2) at (-3.2,0) {};
    \node[nd] (3) at (-2.3,0.7) {};
    \node[nd] (4) at (-2.8,3) {};
    \node[nd] (5) at (0,4.5) {};
    \node[nd] (6) at (0,5.5) {};
    \node[nd] (7) at (2.8,3) {};
    \node[nd] (8) at (2.3,0.7) {};
    \node[nd] (9) at (3.2,0) {};
    
    \draw[->] (1) to[out = -175, in = -60] (2);
    \draw[->] (1) to[out = 175, in = -60] (3);
    \draw[->] (2) to[out=120, in =-120] (4);
    \draw[->] (2) to[out=120,in=180] (6);
    \draw[->] (3) to[out = 120, in = -90] (4);
    
    \draw[->] (3) to[out = 120, in = 180] (5);
    \draw[->] (4) to[out = 60, in = 180] (5);
    \draw[->] (4) to[out = 90, in = 180] (6);
    \draw[->] (5) to[out = 0, in=120] (7);
    \draw[->] (5) to [out = 0, in = 60] (8);
    \draw[->] (6) to[out = 0, in = 90] (7);
    \draw[->] (6) to[out=0,in=60] (9);
    \draw[->] (7) to[out = -90, in = 60] (8);
    \draw[->] (7) to[out = -60, in = 60] (9);
    \draw[->] (8) to[out = -120, in = 5] (1);
    \draw[->] (8) to[out=-120,in=-60] (2);
    \draw[->] (9) to[out = -120, in = -5] (1);
    \draw[->] (9) to[out=-120,in=-60] (3);
 \end{tikzpicture}

%% file: figs/proof1.tex
\begin{tikzpicture}
    
    \node (s1) at (0,1.5) {$(s_1,h)$};
    \node (d1) at (2.5,1.5) {$\dots$};
    \node (d2) at (5.5,1.5) {$\dots$};
    \node (sn) at (8,1.5) {$(s_n,h)$};
    
    \node (s1b) at (0,0) {$\dots$};
    \node (swb) at (2.2,0) {$(s_n,k)$};
    \node (d2b) at (3,0) {$\dots$};
    \node (d3b) at (5,0) {$\dots$};
    \node (sab) at (5.8,0) {$(s_1,k)$};
    \node (snb) at (8,0) {$\dots$};
    
    \draw[->] (s1) to (d1);
    \draw[->] (d1) to (d2);
    \draw[->] (d2) to (sn);
    
    \draw[->] (s1b) to (swb);
    \draw[->] (d2b) to (d3b);
    \draw[->] (sab) to (snb);
    
    \draw[->,densely dotted] (sab) to (s1);
    \draw[->] (s1) to [out = 10, in = 170] (sn);
    \draw[->] (swb) to [out = -10, in = -170] (sab);
 \end{tikzpicture}

%% file: figs/2x3/CO.tex
\begin{tikzpicture}

    \node[inner sep=1pt] (1) at (0,0) {};
    \node[inner sep=1pt] (2) at (1,0) {};
    \node[inner sep=1pt] (3) at (0,1) {};
    \node[inner sep=1pt] (4) at (1,1) {};
    \node[inner sep=1pt] (5) at (0,2) {};
    \node[inner sep=1pt] (6) at (1,2) {};
    
    \draw[->] (1) to (3);
    \draw[->] (1) to[out = 120, in = 240] (5);
    \draw[->] (6) to[out = -60, in = 60] (2);
    \draw[->] (3) to (5);
    \draw[->] (4) to (2);
    \draw[->] (6) to (4);
    
    \draw[->] (6) to (5);
    \draw[->] (3) to (4);
    \draw[->] (1) to (2);
 \end{tikzpicture}

%% file: figs/2x3/MP.tex
\begin{tikzpicture}

    \node[inner sep=1pt] (1) at (0,0) {};
    \node[inner sep=1pt] (2) at (1,0) {};
    \node[inner sep=1pt] (3) at (0,1) {};
    \node[inner sep=1pt] (4) at (1,1) {};
    \node[inner sep=1pt] (5) at (0,2) {};
    \node[inner sep=1pt] (6) at (1,2) {};
    
    \draw[->] (1) to (3);
    \draw[->] (1) to[out = 120, in = 240] (5);
    \draw[->] (6) to[out = -60, in = 60] (2);
    \draw[->] (3) to (5);
    \draw[->] (4) to (2);
    \draw[->] (6) to (4);
    
    \draw[->] (5) to (6);
    \draw[->] (3) to (4);
    \draw[->] (2) to (1);
 \end{tikzpicture}

%% file: figs/2x3/MPCO.tex
\begin{tikzpicture}

    \node[inner sep=1pt] (1) at (0,0) {};
    \node[inner sep=1pt] (2) at (1,0) {};
    \node[inner sep=1pt] (3) at (0,1) {};
    \node[inner sep=1pt] (4) at (1,1) {};
    \node[inner sep=1pt] (5) at (0,2) {};
    \node[inner sep=1pt] (6) at (1,2) {};
    
    \draw[->] (1) to (3);
    \draw[->] (1) to[out = 120, in = 240] (5);
    \draw[->] (6) to[out = -60, in = 60] (2);
    \draw[->] (3) to (5);
    \draw[->] (4) to (2);
    \draw[->] (6) to (4);
    
    \draw[->] (6) to (5);
    \draw[->] (3) to (4);
    \draw[->] (2) to (1);
 \end{tikzpicture}

%% file: figs/2x4/1-0001.tex
\begin{tikzpicture}

    \node[inner sep=1pt] (1) at (0,0) {};
    \node[inner sep=1pt] (2) at (0.8,0) {};
    \node[inner sep=1pt] (3) at (0,0.6) {};
    \node[inner sep=1pt] (4) at (0.8,0.6) {};
    \node[inner sep=1pt] (5) at (0,1.2) {};
    \node[inner sep=1pt] (6) at (0.8,1.2) {};
    \node[inner sep=1pt] (7) at (0,1.8) {};
    \node[inner sep=1pt] (8) at (0.8,1.8) {};
    
    \draw[->] (1) to (3);
    \draw[->] (1) to[out = 120, in = 240] (5);
    \draw[->] (6) to[out = -60, in = 60] (2);
    \draw[->] (3) to (5);
    \draw[->] (4) to (2);
    \draw[->] (6) to (4);
    \draw[->] (5) to (7);
    \draw[->] (8) to (6);
    \draw[->] (3) to[out = 120, in = 240] (7);
    \draw[->] (8) to[out = -60, in = 60] (4);
    \draw[->] (1) to[out = 120, in = 240] (7);
    \draw[->] (8) to[out = -60, in = 60] (2);
    
    \draw[->] (5) to (6);
    \draw[->] (3) to (4);
    \draw[->] (2) to (1);
    \draw[->] (7) to (8);
 \end{tikzpicture}

%% file: figs/2x4/2-0010.tex
\begin{tikzpicture}

    \node[inner sep=1pt] (1) at (0,0) {};
    \node[inner sep=1pt] (2) at (0.8,0) {};
    \node[inner sep=1pt] (3) at (0,0.6) {};
    \node[inner sep=1pt] (4) at (0.8,0.6) {};
    \node[inner sep=1pt] (5) at (0,1.2) {};
    \node[inner sep=1pt] (6) at (0.8,1.2) {};
    \node[inner sep=1pt] (7) at (0,1.8) {};
    \node[inner sep=1pt] (8) at (0.8,1.8) {};
    
    \draw[->] (1) to (3);
    \draw[->] (1) to[out = 120, in = 240] (5);
    \draw[->] (6) to[out = -60, in = 60] (2);
    \draw[->] (3) to (5);
    \draw[->] (4) to (2);
    \draw[->] (6) to (4);
    \draw[->] (5) to (7);
    \draw[->] (8) to (6);
    \draw[->] (3) to[out = 120, in = 240] (7);
    \draw[->] (8) to[out = -60, in = 60] (4);
    \draw[->] (1) to[out = 120, in = 240] (7);
    \draw[->] (8) to[out = -60, in = 60] (2);
    
    \draw[->] (1) to (2);
    \draw[->] (4) to (3);
    \draw[->] (5) to (6);
    \draw[->] (7) to (8);
 \end{tikzpicture}

%% file: figs/2x4/3-0011.tex
\begin{tikzpicture}

    \node[inner sep=1pt] (1) at (0,0) {};
    \node[inner sep=1pt] (2) at (0.8,0) {};
    \node[inner sep=1pt] (3) at (0,0.6) {};
    \node[inner sep=1pt] (4) at (0.8,0.6) {};
    \node[inner sep=1pt] (5) at (0,1.2) {};
    \node[inner sep=1pt] (6) at (0.8,1.2) {};
    \node[inner sep=1pt] (7) at (0,1.8) {};
    \node[inner sep=1pt] (8) at (0.8,1.8) {};
    
    \draw[->] (1) to (3);
    \draw[->] (1) to[out = 120, in = 240] (5);
    \draw[->] (6) to[out = -60, in = 60] (2);
    \draw[->] (3) to (5);
    \draw[->] (4) to (2);
    \draw[->] (6) to (4);
    \draw[->] (5) to (7);
    \draw[->] (8) to (6);
    \draw[->] (3) to[out = 120, in = 240] (7);
    \draw[->] (8) to[out = -60, in = 60] (4);
    \draw[->] (1) to[out = 120, in = 240] (7);
    \draw[->] (8) to[out = -60, in = 60] (2);
    
    \draw[->] (2) to (1);
    \draw[->] (4) to (3);
    \draw[->] (5) to (6);
    \draw[->] (7) to (8);
 \end{tikzpicture}

%% file: figs/2x4/4-0101.tex
\begin{tikzpicture}

    \node[inner sep=1pt] (1) at (0,0) {};
    \node[inner sep=1pt] (2) at (0.8,0) {};
    \node[inner sep=1pt] (3) at (0,0.6) {};
    \node[inner sep=1pt] (4) at (0.8,0.6) {};
    \node[inner sep=1pt] (5) at (0,1.2) {};
    \node[inner sep=1pt] (6) at (0.8,1.2) {};
    \node[inner sep=1pt] (7) at (0,1.8) {};
    \node[inner sep=1pt] (8) at (0.8,1.8) {};
    
    \draw[->] (1) to (3);
    \draw[->] (1) to[out = 120, in = 240] (5);
    \draw[->] (6) to[out = -60, in = 60] (2);
    \draw[->] (3) to (5);
    \draw[->] (4) to (2);
    \draw[->] (6) to (4);
    \draw[->] (5) to (7);
    \draw[->] (8) to (6);
    \draw[->] (3) to[out = 120, in = 240] (7);
    \draw[->] (8) to[out = -60, in = 60] (4);
    \draw[->] (1) to[out = 120, in = 240] (7);
    \draw[->] (8) to[out = -60, in = 60] (2);
    
    \draw[->] (2) to (1);
    \draw[->] (3) to (4);
    \draw[->] (6) to (5);
    \draw[->] (7) to (8);
 \end{tikzpicture}

%% file: figs/2x4/5-0110.tex
\begin{tikzpicture}

    \node[inner sep=1pt] (1) at (0,0) {};
    \node[inner sep=1pt] (2) at (0.8,0) {};
    \node[inner sep=1pt] (3) at (0,0.6) {};
    \node[inner sep=1pt] (4) at (0.8,0.6) {};
    \node[inner sep=1pt] (5) at (0,1.2) {};
    \node[inner sep=1pt] (6) at (0.8,1.2) {};
    \node[inner sep=1pt] (7) at (0,1.8) {};
    \node[inner sep=1pt] (8) at (0.8,1.8) {};
    
    \draw[->] (1) to (3);
    \draw[->] (1) to[out = 120, in = 240] (5);
    \draw[->] (6) to[out = -60, in = 60] (2);
    \draw[->] (3) to (5);
    \draw[->] (4) to (2);
    \draw[->] (6) to (4);
    \draw[->] (5) to (7);
    \draw[->] (8) to (6);
    \draw[->] (3) to[out = 120, in = 240] (7);
    \draw[->] (8) to[out = -60, in = 60] (4);
    \draw[->] (1) to[out = 120, in = 240] (7);
    \draw[->] (8) to[out = -60, in = 60] (2);
    
    \draw[->] (1) to (2);
    \draw[->] (4) to (3);
    \draw[->] (6) to (5);
    \draw[->] (7) to (8);
 \end{tikzpicture}

%% file: figs/2x4/6-1010.tex
\begin{tikzpicture}

    \node[inner sep=1pt] (1) at (0,0) {};
    \node[inner sep=1pt] (2) at (0.8,0) {};
    \node[inner sep=1pt] (3) at (0,0.6) {};
    \node[inner sep=1pt] (4) at (0.8,0.6) {};
    \node[inner sep=1pt] (5) at (0,1.2) {};
    \node[inner sep=1pt] (6) at (0.8,1.2) {};
    \node[inner sep=1pt] (7) at (0,1.8) {};
    \node[inner sep=1pt] (8) at (0.8,1.8) {};
    
    \draw[->] (1) to (3);
    \draw[->] (1) to[out = 120, in = 240] (5);
    \draw[->] (6) to[out = -60, in = 60] (2);
    \draw[->] (3) to (5);
    \draw[->] (4) to (2);
    \draw[->] (6) to (4);
    \draw[->] (5) to (7);
    \draw[->] (8) to (6);
    \draw[->] (3) to[out = 120, in = 240] (7);
    \draw[->] (8) to[out = -60, in = 60] (4);
    \draw[->] (1) to[out = 120, in = 240] (7);
    \draw[->] (8) to[out = -60, in = 60] (2);
    
    \draw[->] (1) to (2);
    \draw[->] (4) to (3);
    \draw[->] (5) to (6);
    \draw[->] (8) to (7);
 \end{tikzpicture}

%% file: figs/2x4/7-1100.tex
\begin{tikzpicture}

    \node[inner sep=1pt] (1) at (0,0) {};
    \node[inner sep=1pt] (2) at (0.8,0) {};
    \node[inner sep=1pt] (3) at (0,0.6) {};
    \node[inner sep=1pt] (4) at (0.8,0.6) {};
    \node[inner sep=1pt] (5) at (0,1.2) {};
    \node[inner sep=1pt] (6) at (0.8,1.2) {};
    \node[inner sep=1pt] (7) at (0,1.8) {};
    \node[inner sep=1pt] (8) at (0.8,1.8) {};
    
    \draw[->] (1) to (3);
    \draw[->] (1) to[out = 120, in = 240] (5);
    \draw[->] (6) to[out = -60, in = 60] (2);
    \draw[->] (3) to (5);
    \draw[->] (4) to (2);
    \draw[->] (6) to (4);
    \draw[->] (5) to (7);
    \draw[->] (8) to (6);
    \draw[->] (3) to[out = 120, in = 240] (7);
    \draw[->] (8) to[out = -60, in = 60] (4);
    \draw[->] (1) to[out = 120, in = 240] (7);
    \draw[->] (8) to[out = -60, in = 60] (2);
    
    \draw[->] (1) to (2);
    \draw[->] (3) to (4);
    \draw[->] (6) to (5);
    \draw[->] (8) to (7);
 \end{tikzpicture}

%% file: figs/2x4/8-0100.tex
\begin{tikzpicture}

    \node[inner sep=1pt] (1) at (0,0) {};
    \node[inner sep=1pt] (2) at (0.8,0) {};
    \node[inner sep=1pt] (3) at (0,0.6) {};
    \node[inner sep=1pt] (4) at (0.8,0.6) {};
    \node[inner sep=1pt] (5) at (0,1.2) {};
    \node[inner sep=1pt] (6) at (0.8,1.2) {};
    \node[inner sep=1pt] (7) at (0,1.8) {};
    \node[inner sep=1pt] (8) at (0.8,1.8) {};
    
    \draw[->] (1) to (3);
    \draw[->] (1) to[out = 120, in = 240] (5);
    \draw[->] (6) to[out = -60, in = 60] (2);
    \draw[->] (3) to (5);
    \draw[->] (4) to (2);
    \draw[->] (6) to (4);
    \draw[->] (5) to (7);
    \draw[->] (8) to (6);
    \draw[->] (3) to[out = 120, in = 240] (7);
    \draw[->] (8) to[out = -60, in = 60] (4);
    \draw[->] (1) to[out = 120, in = 240] (7);
    \draw[->] (8) to[out = -60, in = 60] (2);
    
    \draw[->] (1) to (2);
    \draw[->] (3) to (4);
    \draw[->] (6) to (5);
    \draw[->] (7) to (8);
 \end{tikzpicture}

%% file: figs/2x4/9-1000.tex
\begin{tikzpicture}

    \node[inner sep=1pt] (1) at (0,0) {};
    \node[inner sep=1pt] (2) at (0.8,0) {};
    \node[inner sep=1pt] (3) at (0,0.6) {};
    \node[inner sep=1pt] (4) at (0.8,0.6) {};
    \node[inner sep=1pt] (5) at (0,1.2) {};
    \node[inner sep=1pt] (6) at (0.8,1.2) {};
    \node[inner sep=1pt] (7) at (0,1.8) {};
    \node[inner sep=1pt] (8) at (0.8,1.8) {};
    
    \draw[->] (1) to (3);
    \draw[->] (1) to[out = 120, in = 240] (5);
    \draw[->] (6) to[out = -60, in = 60] (2);
    \draw[->] (3) to (5);
    \draw[->] (4) to (2);
    \draw[->] (6) to (4);
    \draw[->] (5) to (7);
    \draw[->] (8) to (6);
    \draw[->] (3) to[out = 120, in = 240] (7);
    \draw[->] (8) to[out = -60, in = 60] (4);
    \draw[->] (1) to[out = 120, in = 240] (7);
    \draw[->] (8) to[out = -60, in = 60] (2);
    
    \draw[->] (1) to (2);
    \draw[->] (3) to (4);
    \draw[->] (5) to (6);
    \draw[->] (8) to (7);
 \end{tikzpicture}

%% file: figs/3x3/upper_sixloop.tex
\begin{tikzpicture}

    \node[nd] (1) at (0,0) {};
    \node[nd] (2) at (-0.866,0.5) {};
    \node[nd] (3) at (-0.866,1.5) {};
    \node[nd] (5) at (0.866,1.5) {};
    \node[nd] (4) at (0,2) {};
    \node[nd] (6) at (0.866,0.5) {};
    
    \draw[->] (1) to (2);
    \draw[->] (2) to (3);
    \draw[->] (3) to (4);
    \draw[->] (4) to (5);
    \draw[->] (5) to (6);
    \draw[->] (6) to (1);
    
    \node[nd] (7) at (-2,-0.2) {};
    \node[nd] (9) at (2,-0.2) {};
    \node[nd] (8) at (0,3.15) {};

    \draw[->] (7) to (1);
    \draw[->] (7) to (3);
    \draw[->] (7) to[out = -20, in = -110] (6);
    \draw[->] (7) to[out = 90, in = 170] (4);
    
    \draw[->] (8) to (5);
    \draw[->] (8) to (3);
    \draw[->] (8) to[out = -40, in = 60] (6);
    \draw[->] (8) to[out = -140, in = 120] (2);
    
    \draw[->] (9) to (1);
    \draw[->] (9) to (5);
    \draw[->] (9) to[out = -160, in = -70] (2);
    \draw[->] (9) to[out = 90, in = 10] (4);
    
 \end{tikzpicture}

%% file: figs/3x3/lower_sixloop.tex
\begin{tikzpicture}

    \node[nd] (1) at (0,0) {};
    \node[nd] (2) at (-0.866,0.5) {};
    \node[nd] (3) at (-0.866,1.5) {};
    \node[nd] (5) at (0.866,1.5) {};
    \node[nd] (4) at (0,2) {};
    \node[nd] (6) at (0.866,0.5) {};
    
    \draw[->] (1) to (2);
    \draw[->] (2) to (3);
    \draw[->] (3) to (4);
    \draw[->] (4) to (5);
    \draw[->] (5) to (6);
    \draw[->] (6) to (1);
    
    \node[nd] (7) at (-2,-0.2) {};
    \node[nd] (9) at (2,-0.2) {};
    \node[nd] (8) at (0,3.15) {};

    \draw[->] (1) to (7);
    \draw[->] (3) to (7);
    \draw[->] (6) to[in = -20, out = -110] (7);
    \draw[->] (4) to[in = 90, out = 170] (7);
    
    \draw[->] (5) to (8);
    \draw[->] (3) to (8);
    \draw[->] (6) to[in = -40, out = 60] (8);
    \draw[->] (2) to[in = -140, out = 120] (8);
    
    \draw[->] (1) to (9);
    \draw[->] (5) to (9);
    \draw[->] (2) to[in = -160, out = -70] (9);
    \draw[->] (4) to[in = 90, out = 10] (9);
    
 \end{tikzpicture}

%% file: figs/3x3/reflected_sixloop.tex
\begin{tikzpicture}

    \node[nd] (1) at (0,0) {};
    \node[nd] (2) at (-0.866,0.5) {};
    \node[nd] (3) at (-0.866,1.5) {};
    \node[nd] (5) at (0.866,1.5) {};
    \node[nd] (4) at (0,2) {};
    \node[nd] (6) at (0.866,0.5) {};
    
    \draw[->] (1) to (2);
    \draw[->] (3) to (2);
    \draw[->] (3) to (4);
    \draw[->] (5) to (4);
    \draw[->] (5) to (6);
    \draw[->] (1) to (6);
    
    \node[nd] (7) at (-2,-0.2) {};
    \node[nd] (9) at (2,-0.2) {};
    \node[nd] (8) at (0,3.15) {};

    \draw[->] (1) to (7);
    \draw[->] (7) to (3);
    \draw[->] (6) to[in = -20, out = -110] (7);
    \draw[->] (7) to[out = 90, in = 170] (4);
    
    \draw[->] (8) to (5);
    \draw[->] (3) to (8);
    \draw[->] (8) to[out = -40, in = 60] (6);
    \draw[->] (2) to[in = -140, out = 120] (8);
    
    \draw[->] (9) to (1);
    \draw[->] (5) to (9);
    \draw[->] (9) to[out = -160, in = -70] (2);
    \draw[->] (4) to[in = 90, out = 10] (9);
    
 \end{tikzpicture}

%% file: figs/3x3/inner_diamond.tex
\begin{tikzpicture}

    \node[inner sep=1pt] (1) at (0,0) {};
    \node[inner sep=1pt] (2) at (1,0) {};
    \node[inner sep=1pt] (3) at (2,0) {};
    \node[inner sep=1pt] (4) at (0,1) {};
    \node[inner sep=1pt] (5) at (1,1) {};
    \node[inner sep=1pt] (6) at (2,1) {};
    \node[inner sep=1pt] (7) at (0,2) {};
    \node[inner sep=1pt] (8) at (1,2) {};
    \node[inner sep=1pt] (9) at (2,2) {};
    
    \draw[->] (1) to (4);
    \draw[->] (3) to[in = -30, out = 210] (1);
    \draw[->] (9) to[out = -60, in = 60] (3);
    \draw[->] (8) to (9);
    \draw[->] (8) to (5);
    \draw[->] (4) to (5);
    
    \draw[->] (4) to (7);
    \draw[->] (1) to[out = 120, in = -120] (7);
    \draw[->] (7) to (8);
    \draw[->] (7) to[in = 150, out = 30] (9);
    \draw[->] (2) to (1);
    \draw[->] (3) to (2);
    \draw[->] (2) to (5);
    \draw[->] (2) to[out = 60, in = -60] (8);
    \draw[->] (6) to (3);
    \draw[->] (9) to (6);
    \draw[->] (6) to (5);
    \draw[->] (6) to[out = 210, in = -30] (4);
 \end{tikzpicture}

%% file: figs/3x3/outer_diamond.tex
\begin{tikzpicture}

    \node[inner sep=1pt] (1) at (0,0) {};
    \node[inner sep=1pt] (2) at (1,0) {};
    \node[inner sep=1pt] (3) at (2,0) {};
    \node[inner sep=1pt] (4) at (0,1) {};
    \node[inner sep=1pt] (5) at (1,1) {};
    \node[inner sep=1pt] (6) at (2,1) {};
    \node[inner sep=1pt] (7) at (0,2) {};
    \node[inner sep=1pt] (8) at (1,2) {};
    \node[inner sep=1pt] (9) at (2,2) {};
    
    \draw[->] (4) to (1);
    \draw[->] (1) to[out = -30, in = 210] (3);
    \draw[->] (3) to[in = -60, out = 60] (9);
    \draw[->] (9) to (8);
    \draw[->] (5) to (8);
    \draw[->] (5) to (4);
    
    \draw[->] (7) to (4);
    \draw[->] (7) to[out = -120, in = 120] (1);
    \draw[->] (8) to (7);
    \draw[->] (9) to[out = 150, in = 30] (7);
    \draw[->] (1) to (2);
    \draw[->] (2) to (3);
    \draw[->] (5) to (2);
    \draw[->] (8) to[out = -60, in = 60] (2);
    \draw[->] (3) to (6);
    \draw[->] (6) to (9);
    \draw[->] (5) to (6);
    \draw[->] (4) to[in = 210, out = -30] (6);
 \end{tikzpicture}

%% file: main_arxiv.bbl
\begin{thebibliography}{10}
\urlstyle{rm}
\expandafter\ifx\csname url\endcsname\relax
  \def\url#1{\texttt{#1}}\fi
\expandafter\ifx\csname urlprefix\endcsname\relax\def\urlprefix{URL }\fi
\expandafter\ifx\csname doiprefix\endcsname\relax\def\doiprefix{DOI: }\fi
\providecommand{\bibinfo}[2]{#2}
\providecommand{\eprint}[2][]{\url{#2}}

\bibitem{pareto1919manuale}
\bibinfo{author}{Pareto, V.}
\newblock \emph{\bibinfo{title}{Manuale di economia politica: con una
  introduzione alla scienza sociale}}, vol.~\bibinfo{volume}{13}
  (\bibinfo{publisher}{Societ{\`a} editrice libraria}, \bibinfo{year}{1919}).

\bibitem{von2007theory}
\bibinfo{author}{Von~Neumann, J.} \& \bibinfo{author}{Morgenstern, O.}
\newblock \emph{\bibinfo{title}{Theory of games and economic behavior}}
  (\bibinfo{publisher}{Princeton university press}, \bibinfo{year}{1944}).

\bibitem{rasmusen1989games}
\bibinfo{author}{Rasmusen, E.}
\newblock \emph{\bibinfo{title}{Games and Information}},
  vol.~\bibinfo{volume}{13} (\bibinfo{publisher}{Basil Blackwell Oxford},
  \bibinfo{year}{1989}).

\bibitem{myerson1997game}
\bibinfo{author}{Myerson, R.~B.}
\newblock \emph{\bibinfo{title}{Game theory: analysis of conflict}}
  (\bibinfo{publisher}{Harvard university press}, \bibinfo{year}{1997}).

\bibitem{nash1951non}
\bibinfo{author}{Nash, J.}
\newblock \bibinfo{journal}{\bibinfo{title}{Non-cooperative games}}.
\newblock {\emph{\JournalTitle{Annals of mathematics}}}
  \bibinfo{pages}{286--295} (\bibinfo{year}{1951}).

\bibitem{mertens2004ordinality}
\bibinfo{author}{Mertens, J.-F.}
\newblock \bibinfo{journal}{\bibinfo{title}{Ordinality in non cooperative
  games}}.
\newblock {\emph{\JournalTitle{International Journal of Game Theory}}}
  \textbf{\bibinfo{volume}{32}}, \bibinfo{pages}{387--430}
  (\bibinfo{year}{2004}).

\bibitem{durieu2008ordinal}
\bibinfo{author}{Durieu, J.}, \bibinfo{author}{Haller, H.},
  \bibinfo{author}{Qu{\'e}rou, N.} \& \bibinfo{author}{Solal, P.}
\newblock \bibinfo{journal}{\bibinfo{title}{Ordinal games}}.
\newblock {\emph{\JournalTitle{International Game Theory Review}}}
  \textbf{\bibinfo{volume}{10}}, \bibinfo{pages}{177--194}
  (\bibinfo{year}{2008}).

\bibitem{cruz2000ordinal}
\bibinfo{author}{Cruz, J.} \& \bibinfo{author}{Simaan, M.~A.}
\newblock \bibinfo{journal}{\bibinfo{title}{Ordinal games and generalized nash
  and stackelberg solutions}}.
\newblock {\emph{\JournalTitle{Journal of optimization theory and
  applications}}} \textbf{\bibinfo{volume}{107}}, \bibinfo{pages}{205--222}
  (\bibinfo{year}{2000}).

\bibitem{smith1973logic}
\bibinfo{author}{Smith, J.} \& \bibinfo{author}{Price, G.~R.}
\newblock \bibinfo{journal}{\bibinfo{title}{The logic of animal conflict}}.
\newblock {\emph{\JournalTitle{Nature}}} \textbf{\bibinfo{volume}{246}},
  \bibinfo{pages}{15--18} (\bibinfo{year}{1973}).

\bibitem{roughgarden2010algorithmic}
\bibinfo{author}{Roughgarden, T.}
\newblock \bibinfo{journal}{\bibinfo{title}{Algorithmic game theory}}.
\newblock {\emph{\JournalTitle{Communications of the ACM}}}
  \textbf{\bibinfo{volume}{53}}, \bibinfo{pages}{78--86}
  (\bibinfo{year}{2010}).

\bibitem{shoham2008multiagent}
\bibinfo{author}{Shoham, Y.} \& \bibinfo{author}{Leyton-Brown, K.}
\newblock \emph{\bibinfo{title}{Multiagent systems: Algorithmic,
  game-theoretic, and logical foundations}} (\bibinfo{publisher}{Cambridge
  University Press}, \bibinfo{year}{2008}).

\bibitem{daskalakis2009complexity}
\bibinfo{author}{Daskalakis, C.}, \bibinfo{author}{Goldberg, P.~W.} \&
  \bibinfo{author}{Papadimitriou, C.~H.}
\newblock \bibinfo{journal}{\bibinfo{title}{The complexity of computing a nash
  equilibrium}}.
\newblock {\emph{\JournalTitle{SIAM Journal on Computing}}}
  \textbf{\bibinfo{volume}{39}}, \bibinfo{pages}{195--259}
  (\bibinfo{year}{2009}).

\bibitem{chen2006settling}
\bibinfo{author}{Chen, X.} \& \bibinfo{author}{Deng, X.}
\newblock \bibinfo{title}{Settling the complexity of two-player nash
  equilibrium}.
\newblock In \emph{\bibinfo{booktitle}{2006 47th Annual IEEE Symposium on
  Foundations of Computer Science (FOCS'06)}}, \bibinfo{pages}{261--272}
  (\bibinfo{organization}{IEEE}, \bibinfo{year}{2006}).

\bibitem{sandholm2010population}
\bibinfo{author}{Sandholm, W.~H.}
\newblock \emph{\bibinfo{title}{Population games and evolutionary dynamics}}
  (\bibinfo{publisher}{MIT press}, \bibinfo{year}{2010}).

\bibitem{piliouras2014optimization}
\bibinfo{author}{Piliouras, G.} \& \bibinfo{author}{Shamma, J.~S.}
\newblock \bibinfo{title}{Optimization despite chaos: Convex relaxations to
  complex limit sets via poincar{\'e} recurrence}.
\newblock In \emph{\bibinfo{booktitle}{Proceedings of the twenty-fifth annual
  ACM-SIAM Symposium on Discrete Algorithms}}, \bibinfo{pages}{861--873}
  (\bibinfo{organization}{SIAM}, \bibinfo{year}{2014}).

\bibitem{papadimitriou2016nash}
\bibinfo{author}{Papadimitriou, C.} \& \bibinfo{author}{Piliouras, G.}
\newblock \bibinfo{title}{From nash equilibria to chain recurrent sets:
  Solution concepts and topology}.
\newblock In \emph{\bibinfo{booktitle}{Proceedings of the 2016 ACM Conference
  on Innovations in Theoretical Computer Science}}, \bibinfo{pages}{227--235}
  (\bibinfo{year}{2016}).

\bibitem{vlatakis2020no}
\bibinfo{author}{Vlatakis-Gkaragkounis, E.-V.}, \bibinfo{author}{Flokas, L.},
  \bibinfo{author}{Lianeas, T.}, \bibinfo{author}{Mertikopoulos, P.} \&
  \bibinfo{author}{Piliouras, G.}
\newblock \bibinfo{journal}{\bibinfo{title}{No-regret learning and mixed nash
  equilibria: They do not mix}}.
\newblock {\emph{\JournalTitle{Advances in Neural Information Processing
  Systems}}} \textbf{\bibinfo{volume}{33}}, \bibinfo{pages}{1380--1391}
  (\bibinfo{year}{2020}).

\bibitem{benaim2012perturbations}
\bibinfo{author}{Bena{\"\i}m, M.}, \bibinfo{author}{Hofbauer, J.} \&
  \bibinfo{author}{Sorin, S.}
\newblock \bibinfo{journal}{\bibinfo{title}{Perturbations of set-valued
  dynamical systems, with applications to game theory}}.
\newblock {\emph{\JournalTitle{Dynamic Games and Applications}}}
  \textbf{\bibinfo{volume}{2}}, \bibinfo{pages}{195--205}
  (\bibinfo{year}{2012}).

\bibitem{hart2003uncoupled}
\bibinfo{author}{Hart, S.} \& \bibinfo{author}{Mas-Colell, A.}
\newblock \bibinfo{journal}{\bibinfo{title}{Uncoupled dynamics do not lead to
  nash equilibrium}}.
\newblock {\emph{\JournalTitle{American Economic Review}}}
  \textbf{\bibinfo{volume}{93}}, \bibinfo{pages}{1830--1836}
  (\bibinfo{year}{2003}).

\bibitem{papadimitriou2019game}
\bibinfo{author}{Papadimitriou, C.} \& \bibinfo{author}{Piliouras, G.}
\newblock \bibinfo{journal}{\bibinfo{title}{Game dynamics as the meaning of a
  game}}.
\newblock {\emph{\JournalTitle{ACM SIGecom Exchanges}}}
  \textbf{\bibinfo{volume}{16}}, \bibinfo{pages}{53--63}
  (\bibinfo{year}{2019}).

\bibitem{omidshafiei2019alpha}
\bibinfo{author}{Omidshafiei, S.} \emph{et~al.}
\newblock \bibinfo{journal}{\bibinfo{title}{$\alpha$-rank: Multi-agent
  evaluation by evolution}}.
\newblock {\emph{\JournalTitle{Scientific reports}}}
  \textbf{\bibinfo{volume}{9}}, \bibinfo{pages}{1--29} (\bibinfo{year}{2019}).

\bibitem{kleinberg2011beyond}
\bibinfo{author}{Kleinberg, R.~D.}, \bibinfo{author}{Ligett, K.},
  \bibinfo{author}{Piliouras, G.} \& \bibinfo{author}{Tardos, {\'E}.}
\newblock \bibinfo{title}{Beyond the nash equilibrium barrier.}
\newblock In \emph{\bibinfo{booktitle}{ICS}}, \bibinfo{pages}{125--140}
  (\bibinfo{year}{2011}).

\bibitem{cheung2019vortices}
\bibinfo{author}{Cheung, Y.~K.} \& \bibinfo{author}{Piliouras, G.}
\newblock \bibinfo{title}{Vortices instead of equilibria in minmax
  optimization: Chaos and butterfly effects of online learning in zero-sum
  games}.
\newblock In \emph{\bibinfo{booktitle}{Conference on Learning Theory}},
  \bibinfo{pages}{807--834} (\bibinfo{organization}{PMLR},
  \bibinfo{year}{2019}).

\bibitem{silver2016mastering}
\bibinfo{author}{Silver, D.} \emph{et~al.}
\newblock \bibinfo{journal}{\bibinfo{title}{Mastering the game of go with deep
  neural networks and tree search}}.
\newblock {\emph{\JournalTitle{Nature}}} \textbf{\bibinfo{volume}{529}},
  \bibinfo{pages}{484--489} (\bibinfo{year}{2016}).

\bibitem{silver2017mastering}
\bibinfo{author}{Silver, D.} \emph{et~al.}
\newblock \bibinfo{journal}{\bibinfo{title}{Mastering the game of go without
  human knowledge}}.
\newblock {\emph{\JournalTitle{Nature}}} \textbf{\bibinfo{volume}{550}},
  \bibinfo{pages}{354--359} (\bibinfo{year}{2017}).

\bibitem{silver2018general}
\bibinfo{author}{Silver, D.} \emph{et~al.}
\newblock \bibinfo{journal}{\bibinfo{title}{A general reinforcement learning
  algorithm that masters chess, shogi, and go through self-play}}.
\newblock {\emph{\JournalTitle{Science}}} \textbf{\bibinfo{volume}{362}},
  \bibinfo{pages}{1140--1144} (\bibinfo{year}{2018}).

\bibitem{hernandez2019survey}
\bibinfo{author}{Hernandez-Leal, P.}, \bibinfo{author}{Kartal, B.} \&
  \bibinfo{author}{Taylor, M.~E.}
\newblock \bibinfo{journal}{\bibinfo{title}{A survey and critique of multiagent
  deep reinforcement learning}}.
\newblock {\emph{\JournalTitle{Autonomous Agents and Multi-Agent Systems}}}
  \textbf{\bibinfo{volume}{33}}, \bibinfo{pages}{750--797}
  (\bibinfo{year}{2019}).

\bibitem{yang2020overview}
\bibinfo{author}{Yang, Y.} \& \bibinfo{author}{Wang, J.}
\newblock \bibinfo{journal}{\bibinfo{title}{An overview of multi-agent
  reinforcement learning from game theoretical perspective}}.
\newblock {\emph{\JournalTitle{arXiv preprint arXiv:2011.00583}}}
  (\bibinfo{year}{2020}).

\bibitem{candogan2011flows}
\bibinfo{author}{Candogan, O.}, \bibinfo{author}{Menache, I.},
  \bibinfo{author}{Ozdaglar, A.} \& \bibinfo{author}{Parrilo, P.~A.}
\newblock \bibinfo{journal}{\bibinfo{title}{Flows and decompositions of games:
  Harmonic and potential games}}.
\newblock {\emph{\JournalTitle{Mathematics of Operations Research}}}
  \textbf{\bibinfo{volume}{36}}, \bibinfo{pages}{474--503}
  (\bibinfo{year}{2011}).

\bibitem{papadimitriou2018nash}
\bibinfo{author}{Papadimitriou, C.} \& \bibinfo{author}{Piliouras, G.}
\newblock \bibinfo{journal}{\bibinfo{title}{From nash equilibria to chain
  recurrent sets: An algorithmic solution concept for game theory}}.
\newblock {\emph{\JournalTitle{Entropy}}} \textbf{\bibinfo{volume}{20}},
  \bibinfo{pages}{782} (\bibinfo{year}{2018}).

\bibitem{biggarthesis}
\bibinfo{author}{Biggar, O.}
\newblock \emph{\bibinfo{title}{Preference games and sink equilibria}}.
\newblock \bibinfo{type}{{B.Sc. Thesis}}, \bibinfo{school}{Australian National
  University} (\bibinfo{year}{2022}).

\bibitem{conley1978isolated}
\bibinfo{author}{Conley, C.~C.}
\newblock \emph{\bibinfo{title}{Isolated invariant sets and the Morse index}}.
\newblock \bibinfo{number}{38} (\bibinfo{publisher}{American Mathematical
  Soc.}, \bibinfo{year}{1978}).

\bibitem{candogan2013dynamics}
\bibinfo{author}{Candogan, O.}, \bibinfo{author}{Ozdaglar, A.} \&
  \bibinfo{author}{Parrilo, P.~A.}
\newblock \bibinfo{journal}{\bibinfo{title}{Dynamics in near-potential games}}.
\newblock {\emph{\JournalTitle{Games and Economic Behavior}}}
  \textbf{\bibinfo{volume}{82}}, \bibinfo{pages}{66--90}
  (\bibinfo{year}{2013}).

\bibitem{hwang2020strategic}
\bibinfo{author}{Hwang, S.-H.} \& \bibinfo{author}{Rey-Bellet, L.}
\newblock \bibinfo{journal}{\bibinfo{title}{Strategic decompositions of normal
  form games: Zero-sum games and potential games}}.
\newblock {\emph{\JournalTitle{Games and Economic Behavior}}}
  \textbf{\bibinfo{volume}{122}}, \bibinfo{pages}{370--390}
  (\bibinfo{year}{2020}).

\bibitem{diestel2016graph}
\bibinfo{author}{Diestel, R.}
\newblock \bibinfo{journal}{\bibinfo{title}{Graph theory}}.
\newblock {\emph{\JournalTitle{Graduate Texts in Mathematics}}}
  \textbf{\bibinfo{volume}{173}} (\bibinfo{year}{2016}).

\bibitem{omidshafiei2020navigating}
\bibinfo{author}{Omidshafiei, S.} \emph{et~al.}
\newblock \bibinfo{journal}{\bibinfo{title}{Navigating the landscape of
  multiplayer games}}.
\newblock {\emph{\JournalTitle{Nature communications}}}
  \textbf{\bibinfo{volume}{11}}, \bibinfo{pages}{1--17} (\bibinfo{year}{2020}).

\bibitem{goemans2005sink}
\bibinfo{author}{Goemans, M.}, \bibinfo{author}{Mirrokni, V.} \&
  \bibinfo{author}{Vetta, A.}
\newblock \bibinfo{title}{Sink equilibria and convergence}.
\newblock In \emph{\bibinfo{booktitle}{46th Annual IEEE Symposium on
  Foundations of Computer Science (FOCS'05)}}, \bibinfo{pages}{142--151}
  (\bibinfo{organization}{IEEE}, \bibinfo{year}{2005}).

\bibitem{roughgarden2005selfish}
\bibinfo{author}{Roughgarden, T.}
\newblock \emph{\bibinfo{title}{Selfish routing and the price of anarchy}}
  (\bibinfo{publisher}{MIT press}, \bibinfo{year}{2005}).

\bibitem{barany1992classification}
\bibinfo{author}{Barany, I.}, \bibinfo{author}{Lee, J.} \&
  \bibinfo{author}{Shubik, M.}
\newblock \bibinfo{journal}{\bibinfo{title}{Classification of two-person
  ordinal bimatrix games}}.
\newblock {\emph{\JournalTitle{International Journal of Game Theory}}}
  \textbf{\bibinfo{volume}{21}}, \bibinfo{pages}{267--290}
  (\bibinfo{year}{1992}).

\bibitem{morris2004best}
\bibinfo{author}{Morris, S.} \& \bibinfo{author}{Ui, T.}
\newblock \bibinfo{journal}{\bibinfo{title}{Best response equivalence}}.
\newblock {\emph{\JournalTitle{Games and Economic Behavior}}}
  \textbf{\bibinfo{volume}{49}}, \bibinfo{pages}{260--287}
  (\bibinfo{year}{2004}).

\bibitem{moulin1978strategically}
\bibinfo{author}{Moulin, H.} \& \bibinfo{author}{Vial, J.-P.}
\newblock \bibinfo{journal}{\bibinfo{title}{Strategically zero-sum games: the
  class of games whose completely mixed equilibria cannot be improved upon}}.
\newblock {\emph{\JournalTitle{International Journal of Game Theory}}}
  \textbf{\bibinfo{volume}{7}}, \bibinfo{pages}{201--221}
  (\bibinfo{year}{1978}).

\bibitem{monderer1996potential}
\bibinfo{author}{Monderer, D.} \& \bibinfo{author}{Shapley, L.~S.}
\newblock \bibinfo{journal}{\bibinfo{title}{Potential games}}.
\newblock {\emph{\JournalTitle{Games and economic behavior}}}
  \textbf{\bibinfo{volume}{14}}, \bibinfo{pages}{124--143}
  (\bibinfo{year}{1996}).

\bibitem{bang2008digraphs}
\bibinfo{author}{Bang-Jensen, J.} \& \bibinfo{author}{Gutin, G.~Z.}
\newblock \emph{\bibinfo{title}{Digraphs: theory, algorithms and applications}}
  (\bibinfo{publisher}{Springer Science \& Business Media},
  \bibinfo{year}{2008}).

\bibitem{fudenberg1991game}
\bibinfo{author}{Fudenberg, D.} \& \bibinfo{author}{Tirole, J.}
\newblock \emph{\bibinfo{title}{Game Theory}} (\bibinfo{publisher}{MIT press},
  \bibinfo{year}{1991}).

\bibitem{hwang2020simple}
\bibinfo{author}{Hwang, S.-H.} \& \bibinfo{author}{Rey-Bellet, L.}
\newblock \bibinfo{journal}{\bibinfo{title}{Simple characterizations of
  potential games and zero-sum equivalent games}}.
\newblock {\emph{\JournalTitle{Journal of Economic Theory and Econometrics}}}
  \textbf{\bibinfo{volume}{31}}, \bibinfo{pages}{1--13} (\bibinfo{year}{2020}).

\bibitem{hammack2011handbook}
\bibinfo{author}{Hammack, R.~H.}, \bibinfo{author}{Imrich, W.} \&
  \bibinfo{author}{Klav{\v{z}}ar, S.}
\newblock \emph{\bibinfo{title}{Handbook of product graphs}},
  vol.~\bibinfo{volume}{2} (\bibinfo{publisher}{CRC press},
  \bibinfo{year}{2011}).

\bibitem{hubbard2015vector}
\bibinfo{author}{Hubbard, J.~H.} \& \bibinfo{author}{Hubbard, B.~B.}
\newblock \emph{\bibinfo{title}{Vector calculus, linear algebra, and
  differential forms: a unified approach}} (\bibinfo{publisher}{Matrix
  Editions}, \bibinfo{year}{2015}).

\bibitem{alongi2007recurrence}
\bibinfo{author}{Alongi, J.~M.} \& \bibinfo{author}{Nelson, G.~S.}
\newblock \emph{\bibinfo{title}{Recurrence and topology}},
  vol.~\bibinfo{volume}{85} (\bibinfo{publisher}{American Mathematical Soc.},
  \bibinfo{year}{2007}).

\bibitem{hofbauer1996evolutionary}
\bibinfo{author}{Hofbauer, J.}
\newblock \bibinfo{journal}{\bibinfo{title}{Evolutionary dynamics for bimatrix
  games: A hamiltonian system?}}
\newblock {\emph{\JournalTitle{Journal of Mathematical Biology}}}
  \textbf{\bibinfo{volume}{34}}, \bibinfo{pages}{675--688}
  (\bibinfo{year}{1996}).

\bibitem{balduzzi2018mechanics}
\bibinfo{author}{Balduzzi, D.} \emph{et~al.}
\newblock \bibinfo{title}{The mechanics of n-player differentiable games}.
\newblock In \emph{\bibinfo{booktitle}{International Conference on Machine
  Learning}}, \bibinfo{pages}{354--363} (\bibinfo{organization}{PMLR},
  \bibinfo{year}{2018}).

\bibitem{quint1997theorem}
\bibinfo{author}{Quint, T.} \& \bibinfo{author}{Shubik, M.}
\newblock \bibinfo{journal}{\bibinfo{title}{A theorem on the number of nash
  equilibria in a bimatrix game}}.
\newblock {\emph{\JournalTitle{International Journal of Game Theory}}}
  \textbf{\bibinfo{volume}{26}}, \bibinfo{pages}{353--359}
  (\bibinfo{year}{1997}).

\bibitem{harsanyi1988general}
\bibinfo{author}{Harsanyi, J.~C.}, \bibinfo{author}{Selten, R.} \emph{et~al.}
\newblock \bibinfo{journal}{\bibinfo{title}{A general theory of equilibrium
  selection in games}}.
\newblock {\emph{\JournalTitle{MIT Press Books}}} \textbf{\bibinfo{volume}{1}}
  (\bibinfo{year}{1988}).

\bibitem{li2020verification}
\bibinfo{author}{Li, C.}, \bibinfo{author}{He, F.} \& \bibinfo{author}{Hao, N.}
\newblock \bibinfo{journal}{\bibinfo{title}{Verification and design of zero-sum
  potential games}}.
\newblock {\emph{\JournalTitle{IFAC-PapersOnLine}}}
  \textbf{\bibinfo{volume}{53}}, \bibinfo{pages}{16932--16937}
  (\bibinfo{year}{2020}).

\end{thebibliography}
